%% file: main.tex
\documentclass[sigplan,screen]{acmart}
\newif\iffullversion\fullversiontrue
\iffullversion
\renewcommand\footnotetextcopyrightpermission[1]{} %
\fi
\usepackage[utf8]{inputenc}
\usepackage[T1]{fontenc}
\usepackage{microtype}
\usepackage{amsmath, amssymb,amsthm}
\usepackage{mathtools}
\usepackage{paralist}
\usepackage{hyperref}
\usepackage{listings}
\usepackage{mathpartir} 
\usepackage{tikz}
\usepackage{stmaryrd}
\usepackage{xspace}
\usepackage{dsfont}
\usepackage{bbold}
\usepackage{setspace}
\usepackage{todonotes}
\usepackage{declarations}
\usepackage{balance}
\usepackage{subcaption}
\usepackage{comment}
\usepackage{color}
\usepackage{xr}
\usepackage{enumitem}
\usepackage{makecell}
\usepackage{footnote}
\iffullversion
\usepackage{flushend}
\fi

\newcommand{\abs}[1]{\lvert#1\rvert}%

\newcommand{\RR}{\mathbb{R}}

\graphicspath{ {figures/} }

\lstdefinelanguage{z3}{
  keywords={declare, const, assert, Int, Real, and, or, forall, minimize, maximize, check, sat, get, model},
}

\lstdefinelanguage{WHILE}{
  keywords={Lap, while, if, then, else, cons, havoc},
}

\hypersetup{
  hidelinks,
  colorlinks=true,
  linkcolor=black,
  citecolor=black,
  urlcolor=black,
}

\let\origthelstnumber\thelstnumber
\makeatletter
\newcommand*\Suppressnumber{%
  \lst@AddToHook{OnNewLine}{%
    \let\thelstnumber\relax%
     \advance\c@lstnumber-\@ne\relax%
    }%
}

\def\thickhline{%
  \noalign{\ifnum0=`}\fi\hrule \@height \thickarrayrulewidth \futurelet
   \reserved@a\@xthickhline}
\def\@xthickhline{\ifx\reserved@a\thickhline
               \vskip\doublerulesep
               \vskip-\thickarrayrulewidth
             \fi
      \ifnum0=`{\fi}}
      
\newcommand*\Reactivatenumber{%
  \lst@AddToHook{OnNewLine}{%
   \let\thelstnumber\origthelstnumber%
   \advance\c@lstnumber\@ne\relax}%
}
\makeatother

\newcommand\laplace{\ensuremath{\cod{Lap}}\xspace}

\newcommand\lapm[1]{\laplace~{#1}}

\newcommand\op{\cod{op}}
\newcommand\tylist{\cod{list}}

\newcommand\kout{\cod{return}}

\newcommand\restrict[3]{{#1}\upharpoonright^{#3}_{#2}}

\newcommand\outcmd[1]{\kout~{#1}}
\newcommand\havoc[1]{\cod{havoc}~{#1}}

\newcommand\instrument[1]{\underline{\ensuremath{#1}}}
\newcommand\annotation[1]{\colorbox{lightgray}{\ensuremath{#1}}}
\newcommand\tT{{\tilde T}}

\newcommand\pre{{\Psi}}

\newcommand\store{{m}}
\newcommand\aligndstore[2]{{#1^\alignd}{#2}}
\newcommand\shadowstore[2]{{#1^\shadow}{#2}}
\newcommand\Store{\mathcal{M}}
\newcommand\distance[1]{\widehat{#1}}
\DeclarePairedDelimiter\set\{\}
\newcommand\amem[2]{\first{#1}{#2}}    %
\newcommand\smem[2]{\second{#1}{#2}}   %
\newcommand\cmdfull[3]{{#1}\set{#3}{#2}}

\newcommand\vpriv[1]{{\mathbf{v}_{{#1}}}}

\newcommand\evalexpr[2]{\trans{#1}_{#2}}
\newcommand\real{{r}}
\newcommand\bexpr{{\mathbb{b}}}
\newcommand\dexpr{{\mathbb{d}}}
\newcommand\nexpr{{\mathbb{n}}}

\newcommand\transform{\rightharpoonup}

\newcommand\priv{\ensuremath{\epsilon}}

\newcommand\dist{\mathbf{Dist}}
\newcommand\extmemfull[3]{{#1}\uplus({#2})} %

\newcommand\support{\cod{support}}
\newcommand\unitop{\cod{unit}}
\newcommand\bindop{\cod{bind}}

\newcommand\lang{ShadowDP\xspace}
\newcommand\basety{\mathcal{B}}
\newcommand\dgdist{\ensuremath{\mathds{1}}\xspace}

\newcommand\funsigfour[4]{
\makebox[1.3cm]{\textbf{function}} \textsc{#1}~(#2) \par
\makebox[1.3cm]{}\makebox[1.2cm]{\textbf{returns}}{#3}\par
\makebox[2.0cm]{\textbf{precondition}}{#4}\par}

\newif\ifblinded\blindedtrue
\newif\ifreport\reportfalse
\blindedfalse

\iffullversion
\newcommand\appendixref{Appendix}
\else
\newcommand\appendixref{full version of this paper~\cite{shadowdpreport}}
\fi

\setcopyright{acmcopyright}
\acmPrice{15.00}
\acmDOI{10.1145/3314221.3314619}
\acmYear{2019}
\copyrightyear{2019}
\acmISBN{978-1-4503-6712-7/19/06}
\acmConference[PLDI '19]{Proceedings of the 40th ACM SIGPLAN Conference on Programming Language Design and Implementation}{June 22--26, 2019}{Phoenix, AZ, USA}
\acmBooktitle{Proceedings of the 40th ACM SIGPLAN Conference on Programming Language Design and Implementation (PLDI '19), June 22--26, 2019, Phoenix, AZ, USA}



\title{Proving Differential Privacy with Shadow Execution}
\ifblinded
\authorinfo{Anonymized for review}{}{}
\else
\author{Yuxin Wang}
\affiliation{%
  \institution{Pennsylvania State University}
  \city{University Park}
  \state{PA}
  \postcode{16802}
  \country{USA}
}
\email{yxwang@psu.edu}
\author{Zeyu Ding}
\affiliation{%
  \institution{Pennsylvania State University}
  \city{University Park}
  \state{PA}
  \postcode{16802}
  \country{USA}
}
\email{zyding@psu.edu}
\author{Guanhong Wang}
\affiliation{%
  \institution{Pennsylvania State University}
  \city{University Park}
  \state{PA}
  \postcode{16802}
  \country{USA}
}
\email{gpw5092@psu.edu}
\author{Daniel Kifer}
\affiliation{%
  \institution{Pennsylvania State University}
  \city{University Park}
  \state{PA}
  \postcode{16802}
  \country{USA}
}
\email{dkifer@cse.psu.edu}
\author{Danfeng Zhang}
\affiliation{%
  \institution{Pennsylvania State University}
  \city{University Park}
  \state{PA}
  \postcode{16802}
  \country{USA}
}
\email{zhang@cse.psu.edu}
\fi

\begin{document}

\setlist[enumerate]{leftmargin=0.5cm}%
\setlist[itemize]{leftmargin=0.5cm}%

\begin{abstract}
Recent work on formal verification of differential privacy shows a trend toward usability and expressiveness -- generating a correctness proof of sophisticated algorithm while minimizing the annotation burden on programmers. Sometimes, combining those two requires substantial changes to program logics: one recent paper is able to verify Report Noisy Max automatically, but it involves a complex verification system using customized program logics and verifiers.

In this paper, we propose a new proof technique, called shadow execution, and embed it into a language called ShadowDP. ShadowDP uses shadow execution to generate proofs of differential privacy with very few programmer annotations and without relying on customized logics and verifiers. In addition to verifying Report Noisy Max, we show that it can verify a new variant of Sparse Vector that reports the gap between some noisy query answers and the noisy threshold. Moreover, ShadowDP reduces the complexity of verification: for all of the algorithms we have evaluated, type checking and verification in total takes at most 3 seconds, while prior work takes minutes on the same algorithms.
\end{abstract}

\maketitle

\section{Introduction}\label{sec:intro}
\input{introduction}

\section{Preliminaries and Illustrating Example}\label{sec:preliminaries}\label{sec:illustrative}
\input{preliminaries}

\section{\lang: Syntax and Semantics}\label{sec:langzero}
\input{lang}

\section{Soundness}\label{sec:soundness}
\input{soundness}

\section{Implementation and Evaluation}\label{sec:impexp}
\input{evaluation}

\section{Related Work}\label{sec:related}
\input{relatedwork}

\section{Conclusions and Future Work}\label{sec:conclusions}
\input{conclusions}

\ifblinded
\else
\section*{Acknowledgments}
We thank our shepherd Dana Drachsler-Cohen and anonymous PLDI reviewers for their helpful
suggestions. This work is funded by NSF awards \#1228669, \#1702760, \#1816282
and \#1566411.
\fi

\bibliographystyle{ACM-Reference-Format}
\iffullversion
\balance
\fi
\bibliography{diffpriv}
\iffullversion
\input{appendix}
\fi
\end{document}

%% file: introduction.tex
Differential privacy is increasingly being used in industry
\cite{rappor,appledpscale,elasticsensitivity} and government agencies
\cite{abowd} to provide statistical information about groups of people without
violating their privacy. Due to the prevalence of errors in published
algorithms and code \cite{ninghuisparse}, formal verification of differential
privacy is critical to its success. 

The initial line of work on formal verification for differential privacy (e.g.,
\cite{Barthe15,BartheICALP2013,Barthe12,Barthe14,Barthe16}) was concerned with
increasing expressiveness. A parallel line of work (e.g.,
\cite{pinq,Roy10Airavat,gupt,ectelo}) focuses more on usability -- on
developing platforms that keep track of the privacy cost of an algorithm while
limiting the types of algorithms that users can produce.

A recent line of work (most notably LightDP~\cite{lightdp} and Synthesizing
Coupling Proofs~\cite{Aws:synthesis}) has sought to combine expressiveness and
usability by providing verification tools that infer most (if not all) of the
proof of privacy. The benchmark algorithms for this task were Sparse Vector
\cite{ninghuisparse,diffpbook} and Report Noisy Max \cite{diffpbook}. LightDP
\cite{lightdp} was the first system that could verify Sparse Vector with very
few annotations, but it could not verify tight privacy bounds on
Report Noisy Max \cite{diffpbook}. It is believed that proofs using
\emph{randomness alignment}, the proof technique that underpins LightDP, are
often simpler, while \emph{approximate coupling}, the proof technique that
underpins \cite{Barthe15,BartheICALP2013,Barthe12,Barthe14,Barthe16}, seems to
be more expressive \cite{Aws:synthesis}.  Subsequently, Albarghouthi and Hsu
\cite{Aws:synthesis} produced the first fully automated system that verifies
both Sparse Vector and Report Noisy Max. Although this new system takes inspiration from randomness alignment to
simplify approximate coupling proofs, its verification system
still involves challenging features such as first-order Horn clauses and
probabilistic constraints; it takes minutes to  verify simple algorithms. The
complex verification system also prevents it from reusing off-the-shelf
verification tools.

In this paper, we present ShadowDP, a language for verifying differentially
private algorithms. It is based on a new proof technique called ``shadow
execution'', which enables language-based proofs based on standard program
logics. Built on randomness alignment, it transforms a probabilistic program
into a program in which the privacy cost is explicit; so that the target
program can be readily verified by off-the-shelf verification tools.  However,
unlike LightDP, it can verify more challenging algorithms such as
Report Noisy Max and a novel variant of Sparse Vector called Difference Sparse
Vector. We show that with minimum annotations, challenging
algorithms that took minutes to verify by~\cite{Aws:synthesis} (excluding proof
synthesis time) now can be verified within 3 seconds with an off-the-shelf
model checker. 

One extra benefit of this approach based on randomness alignment is that the
transformed program can also be analyzed by standard symbolic executors. This
appears to be an important property in light of recent work on detecting
counterexamples for buggy programs
\cite{Ding2018CCS,Bichsel2018CCS,diffproptest,FarinaRelational}. Producing a
transformed program that can be used for verification of correct programs and
bug-finding for incorrect programs is one aspect that is of independent
interest (however, we leave this application of transformed programs to future
work).

In summary, this paper makes the following contributions:
\begin{enumerate}
\item Shadow execution, a new proof technique for differential privacy
(Section~\ref{sec:overview}).

\item \lang, a new imperative language (Section~\ref{sec:langzero}) with a
flow-sensitive type system (Section~\ref{sec:typing}) for verifying sophisticated privacy-preserving
algorithms.

\item A formal proof that the verification of the transformed program by \lang
implies that the source code is $\priv$-differentially private
(Section~\ref{sec:soundness}).

\item Case studies on sophisticated  algorithms showing that verifying
privacy-preserving algorithms using \lang requires little programmer annotation
burden and verification time (Section~\ref{sec:impexp}).

\item Verification of a variant of Sparse Vector Technique that releases the
difference between noisy query answers and a noisy threshold at the same
privacy level as the original algorithm
\cite{ninghuisparse,diffpbook}. To the best of our knowledge, this variant has
not been studied before.
\end{enumerate}

%% file: preliminaries.tex
\subsection{Differential Privacy}\label{sec:differential_privacy}

Differential privacy is now considered a gold standard in privacy protections after recent high profile adoptions \cite{rappor,appledpscale,elasticsensitivity,abowd}.
There are currently several popular variants of differential privacy \cite{dwork06Calibrating,dworkKMM06:ourdata,BS2016:zcdp,M2017:Renyi}. In this paper, we focus
on the verification of algorithms that satisfy pure differential
privacy~\cite{dwork06Calibrating}, which has several key advantages -- it is the strongest one among them, the 
most popular one, and the easiest to explain to non-technical end-users \cite{dpprimer}.

Differential privacy requires an algorithm to inject carefully calibrated
random noise during its computation. The purpose of the noise is to hide the effect of any person's record on the output
of the algorithm.
In order to present the formal definition, we first define
the set of \emph{sub-distributions} over a discrete set $A$, written
$\dist(A)$, as the set of functions $\mu: A \rightarrow [0,1]$, such that
$\sum_{a\in A} \mu(a)\leq 1$. When applied to an event $E\subseteq A$, we
define $\mu(E)\defn \sum_{e\in E} \mu(e)$.~\footnote{As is standard in this
line of work (e.g.,~\cite{Barthe16,lightdp}), we assume a sub-distribution
instead of a distribution, since sub-distribution gives rise to an elegant
program semantics in face of non-terminating programs~\cite{Kozen81}.}

Differential privacy relies on the notion of adjacent databases (e.g., pairs of databases that differ on one record). Since
differentially-private algorithms sometimes operate on query results from
databases, we abstract adjacent databases as an adjacency relation
$\pre\subseteq A\times A$ on input query answers. For differential privacy, the
most commonly used relations are: 
\begin{inparaenum}[(1)]
\item each query answer may differ by at most $n$ (for some number $n$), and
\item at most one query answer may differ, and that query answer differs by at
most $n$. This is also known as \emph{sensitivity} of the queries.
\end{inparaenum}

\begin{definition}[Pure Differential privacy]
\label{def:diffpriv}
Let $\priv\geq 0$. A probabilistic computation $M: A \rightarrow B$ is
$\priv$-differentially private with respect to an adjacency relation
$\pre$ if for every pair of inputs $a_1,a_2\in A$ such that $a_1 \pre a_2$, and
every output subset $E\subseteq B$, 
\vspace{-1ex}
\[P(M (a_1)\in E)\leq e^\priv P(M (a_2)\in E).\]
\vspace{-4.5ex}
\end{definition}

\vspace{-2ex}
\subsection{Randomness Alignment}
\label{sec:randomness_alignment}
\emph{Randomness Alignment}~\cite{lightdp} is a simple yet powerful technique to prove differential privacy. Here, we illustrate the key idea
with a fundamental mechanism for satisfying differential privacy--the
\emph{Laplace Mechanism}~\cite{exponentialMechanism}. 

Following the notations in Section~\ref{sec:differential_privacy}, we consider an arbitrary pair of query answers $a_1$ and $a_2$ that differ by at most 1, i.e., $-1\le a_1 - a_2 = c \le 1$. The Laplace Mechanism (denoted as $M$) simply releases $a+\eta$, where
$\eta$ is a random noise sampled from the Laplace distribution of mean 0 and
scale $1/\epsilon$; we use $p_{1/\epsilon}$ to denote its density function.
The goal of randomness alignment is to ``align'' the random noise in two
executions $M(a_1)$ and $M(a_2)$, such that $M(a_1)=M(a_2)$, with a
corresponding privacy cost. To do so, we create an \emph{injective} function
$f:\RR\rightarrow \RR$ that maps $\eta$ to $\eta+c$. Obviously, $f$ is an
alignment since $a_1+\eta=a_2+f(\eta)$ for any $a_1$, $a_2$.  Then for an
arbitrary set of outputs $E\subseteq \RR$, we have:
\vspace{-1ex}
\begin{align*}
    P(M(a_1)\in E) &=  \sum_{\eta \mid a_1 + \eta \in E} p_{1/\epsilon} (\eta)
                   \leq  \sum_{\eta \mid a_2+f(\eta) \in E} p_{1/\epsilon} (\eta)\\
                   &\leq e^{\epsilon} \sum_{\eta \mid a_2+f(\eta) \in E} p_{1/\epsilon} (f(\eta)) \\
                      &= e^{\epsilon} \sum_{\eta \mid a_2+\eta \in E} p_{1/\epsilon} (\eta) = e^\epsilon P(M(a_2)\in E)
\end{align*}
The first inequality is by the definition of $f$: $a_1+\eta \in E \implies a_2+f(\eta) \in E$. The $e^\epsilon$ factor results from the fact that $p_{1/\epsilon} (\eta+c) \big/ p_{1/\epsilon} (\eta) \leq e^{\abs{c}\cdot \epsilon} \leq e^\epsilon$, when the Laplace distribution has scale $1/\epsilon$. The second to last equality is by change of variable from $f(\eta)$ to $\eta$ in the summation, using the injectivity of $f$.

In general, let $H\in\RR^n$ be the random noise vector that a mechanism $M$
uses. A randomness alignment for $a_1\Psi a_2$ is a function $f:\RR^n
\rightarrow \RR^n$ such that: \begin{enumerate}[leftmargin=5mm]
    \item $M(a_2)$ with noise $f(H)$ outputs the same result as $M(a_1)$ with noise $H$  (hence the name Randomness Alignment). 
    \item $f$ is injective (this is to allow change of variables).
\end{enumerate}

\subsection{The Report Noisy Max Algorithm}

To illustrate the challenges in proving differential privacy, we consider the
Report Noisy Max algorithm~\cite{diffpbook}, whose source code is shown on the top of
Figure~\ref{alg:transformed_noisymax}. It can be used as a building block in
algorithms that iteratively generate differentially private synthetic data by finding (with high probability) 
the identity of the query for which the synthetic data currently has the largest error~\cite{mwem}.

\begin{figure}
\small
\setstretch{0.9}
\raggedright
\noindent\rule{\linewidth}{2\arrayrulewidth}
\funsigfour{NoisyMax}{$\priv$, $ \cod{size}\annotation{:\tyreal_{\pair 0 0}}$; $\cod{q}\annotation{:\tylist~\tyreal_{\pair * *}}$}{$\cod{max}\annotation{:\tyreal_{\pair 0 *}}$}{$\alldiffer$}
\algrule
\begin{lstlisting}[frame=none, escapechar=@]
i := 0; bq := 0; max := 0;
while (i < size)
   $\eta\text{ := }\lapm(2/\priv)\sampleannotation{\Omega\mathbin{?}\shadow:\alignd}{\Omega\mathbin{?}2:0}$;
   if (q[i] + $\eta$ > bq $\lor$ i = 0)@\label{line:noisymax_branch}@
     max := i;@\label{line:noisymax_true_branch}@
     bq := q[i] + $\eta$;@\label{line:noisymax_true_branch_2}@
   i := i + 1;
\end{lstlisting}
The transformed program (slightly simplified for readability), where underlined commands are added by the type system:
\begin{lstlisting}[frame=none]
$\instrument{\vpriv{\epsilon}\text{ := 0;}}$ /*\quad\first{\distance{\eta}}\text{ := 0;}}*/
i := 0; bq := 0; max := 0;
$\instrument{\first{\distance{\cod{bq}}}\text{ := 0;}\quad\second{\distance{\cod{bq}}}\text{ := 0;}}$ /*\quad\second{\distance{\cod{max}}}\text{ := 0;}*/
while (i < size)
   $\instrument{\assertcmd(\text{i < size})\text{;}}$
   $\instrument{\havoc{\eta}\text{; }\vpriv{\epsilon}\text{ := }\Omega\mathbin{?}(\text{0 + } \priv):(\vpriv{\epsilon}\text{ + 0})\text{;}}$ 
   if (q[i] + $\eta$ > bq $\lor$ i = 0) 
     $\instrument{\assertcmd(\text{q[i]} + \first{\distance{\text{q}}}\text{[i]}+\eta +2>\text{bq} + \second{\distance{\text{bq}}}\mathbin{\lor}\text{i = 0});}$/*\instrument{\second{\distance{\text{max}}}\text{ := max + }\second{\distance{\text{max}}}\text{ - i;}}*/
     max := i;
     $\instrument{\second{\distance{\text{bq}}}\text{ := bq + }\second{\distance{\text{bq}}}\text{ - (q[i] + }\eta\text{);}}$
     bq := q[i] + $\eta$;
     $\instrument{\first{\distance{\text{bq}}}\text{ := }\first{\distance{\text{q}}}\text{[i]} + 2\text{;}}$ /*\quad\first{\distance{\eta}}\text{ := 2;}*/
   $\instrument{\cod{\textbf{else}}}$
     $\instrument{\assertcmd(\lnot (\text{q[i]}+\first{\distance{\text{q}}}\text{[i]}+\eta+0 > \text{bq} + \first{\distance{\text{bq}}}\mathbin{\lor}\text{i = 0}));}$/*\instrument{\first{\distance{\text{bq}}}\text{ := }\first{\distance{\text{bq}}}\text{;}}\quad\first{\distance{\eta}}\text{ := 0;}*/
   // shadow execution
   $\instrument{\cod{\textbf{if}}~(\text{q[i]} + \second{\distance{\text{q}}}\text{[i]} + \eta> \text{bq} + \second{\distance{\text{bq}}}\mathbin{\lor}\text{i = 0})}$/*\instrument{\second{\distance{\text{max}}}\text{ := i - max;}}*/
     $\instrument{\second{\distance{\text{bq}}}\text{ := q[i] + }\second{\distance{\text{q}}}\text{[i]} + \eta - \text{bq}\text{;}}$
   i := i + 1;
\end{lstlisting}
\vspace{-2ex}
\noindent\rule{\linewidth}{2\arrayrulewidth}
\vspace{-4ex}
\caption{Verifying Report Noisy Max with \lang. Here, $q$ is a list of query answers
from a database, and max is the query index of the maximum query with Laplace
noise generated at line 3. To verify the algorithm on the top, a programmer
provides function specification as well as annotation for sampling command (annotations are shown in gray, where $\Omega$ represents the branch condition).
\lang checks the source code and generates the transformed
code (at the bottom), which can be verified with off-the-shelf verifiers.}
	\vspace{-2ex}
\label{alg:transformed_noisymax}
\end{figure}

The algorithm takes a list $q$ of query answers, each of which differs by at
most 1 if the underlying database is replaced with an adjacent one.  The goal
is to return the index of the largest query answer (as accurately as possible
subject to privacy constraints).

To achieve differential privacy, the algorithm adds appropriate Laplace noise
to each query. Here, $\lapm(2/\priv)$ draws one sample from the Laplace
distribution with mean zero and a scale factor $(2/\priv)$.  For privacy, the
algorithm uses the noisy query answer ($q[i]+\eta$) rather than the true query
answer ($q[i]$) to compute and return the \emph{index} of the maximum (noisy) query
answer. Note that the return value is listed right below the function signature in the source code.

\paragraph{Informal proof using randomness alignment}
Proofs of correctness of Report Noisy Max can be found in~\cite{diffpbook}.  We will
start with an informal correctness argument, based on the \emph{randomness
alignment} technique (Section~\ref{sec:randomness_alignment}), to illustrate
subtleties involved in the proof.

Consider the following two databases $D_1, D_2$ that differ on one record, and
their corresponding query answers:
\begin{align*}
D_1:\quad q[0]=1,\quad q[1]=2,\quad q[2]=2\\
D_2:\quad q[0]=2,\quad q[1]=1,\quad q[2]=2
\end{align*}
Suppose in one execution on $D_1$, the noise added to $q[0]$, $q[1]$, $q[2]$ is
$\alpha^{(1)}_0=1$, $\alpha^{(1)}_1=2$, $\alpha^{(1)}_2=1$, respectively. In
this case, the noisy query answers are $q[0]+\alpha^{(1)}_0=2$, 
$q[1]+\alpha^{(1)}_1=4$, $q[2]+\alpha^{(1)}_2=3$ and so the algorithm returns
$1$, which is the index of the maximum noise query answer of $4$.

\paragraph{Aligning randomness}
\label{sec:aligning_randomness}
As shown in Section~\ref{sec:randomness_alignment}, we need to create an
injective function of random bits in $D_1$ to random bits in
$D_2$ to make the output the same. Recall that
$\alpha^{(1)}_0,\alpha^{(1)}_1,\alpha^{(1)}_2$ are the noise added to $D_1$,
now let $\alpha^{(2)}_0,\alpha^{(2)}_1,\alpha^{(2)}_2$ be the noise added to
the queries $q[0],q[1],q[2]$ in $D_2$, respectively. Consider the following injective
function: for any query except for $q[1]$, use the
same noise as on $D_1$;
add 2 to the noise used for $q[1]$ (i.e., $\alpha^{(2)}_1 = \alpha^{(1)}_1 +2$).

In our running example, execution on $D_2$ with this alignment function would
result in noisy query answers $q[0]+\alpha^{(2)}_0=3$, $q[1]+\alpha^{(2)}_1=5$,
$q[2]+\alpha^{(2)}_2=3$. Hence, the output once again is $1$, the index of
query answer $5$. 

In fact, we can prove that under this alignment,
\emph{every execution on $D_1$ where 1 is returned} would result in an
execution on $D_2$ that produces the same answer due to two facts:

\begin{enumerate}[leftmargin=5mm]
\item On $D_1$, $q[1]+\alpha^{(1)}_1$ has the maximum value.

\item On $D_2$, $q[1]+\alpha^{(2)}_1$ is greater than $q[1]+\alpha^{(1)}_1+1$
on $D_1$ due to $\alpha^{(2)}_1=\alpha^{(1)}_1+2$ and the adjacency assumption.
\end{enumerate}

Hence, $q[1]+\alpha^{(2)}_1$ on $D_2$ is greater than $q[i]+\alpha^{(1)}_i+1$
on $D_1$ for any $i$. By the adjacency assumption, that is the same as
$q[1]+\alpha^{(2)}_1$ is greater than any $q[i]+\alpha^{(2)}_i$ on $D_2$.
Hence, based on the same argument in Section~\ref{sec:randomness_alignment}, we
can prove that the Report Noisy Max algorithm is $\epsilon$-private.

\paragraph{Challenges} 
Unfortunately, the alignment function above only applies to executions on $D_1$
where index $1$ is returned. If there is one more query $q[3]=4$ and the
execution gets noise $\alpha^{(1)}_3=1$ for that query, the execution on $D_1$
will return index 3 instead of 1. To align randomness on $D_2$, we need to
construct a different alignment function (following the construction above)
that adds noise in the following way: for any query except for $q[3]$, use the
same noise as on $D_1$; add 2 to the noise used for $q[3]$ (i.e.,
$\alpha^{(2)}_3 = \alpha^{(1)}_3 +2$). In other words, to carry out the proof,
the alignment for each query depends on the queries and noise yet to happen
\emph{in the future}.

One approach of tackling this challenge, followed by existing language-based
proofs of Report Noisy Max~\cite{Barthe16,Aws:synthesis}, is to use the pointwise
lifting argument: informally, if we can show that for any value $i$, execution
on $D_1$ returns value $i$ implies execution on $D_2$ returns value $i$ (with a
privacy cost bounded by $\epsilon$), then a program is $\epsilon$-differential
private. However, this argument does not apply to the randomness alignment
technique. Moreover, doing so requires a customized program logic for proving differential privacy.

\subsection{Approach Overview} 
\label{sec:overview}
In this paper, we propose a new proof technique ``shadow execution'',
which enables language-based proofs based on \emph{standard} program logics.
The key insight is to track a \emph{shadow execution} on $D_2$ where the
\emph{same noise} is always used as on $D_1$. For our running example, we illustrate
the shadow execution in Figure~\ref{fig:align}, with random noise
$\alpha^{(\shadow)}_0$, $\alpha^{(\shadow)}_1$ and so on. Note that the shadow
execution uses $\alpha^{(\shadow)}_i=\alpha^{(1)}_i$ for all $i$.

With the shadow execution, we can construct a randomness alignment for each
query $i$ as follows:
\begin{enumerate}[label=Case \arabic*:,leftmargin=*]
\item Whenever $q[i]+\alpha^{(1)}_i$ is the maximum value so far on $D_1$
(i.e., $max$ is updated), we use the alignments on \emph{shadow execution} for all
previous queries but a noise $\alpha^{(1)}_i+2$ for $q[i]$ on $D_2$. 

\item Whenever $q[i]+\alpha^{(1)}_i$ is smaller than or equal to any previous
noise query answer (i.e., $max$ is not updated), we keep the previous
alignments for previous queries and use noise $\alpha^{(1)}_i$ for $q[i]$ on
$D_2$.
\end{enumerate}

\begin{figure}
\centering
\includegraphics[width=.9\columnwidth]{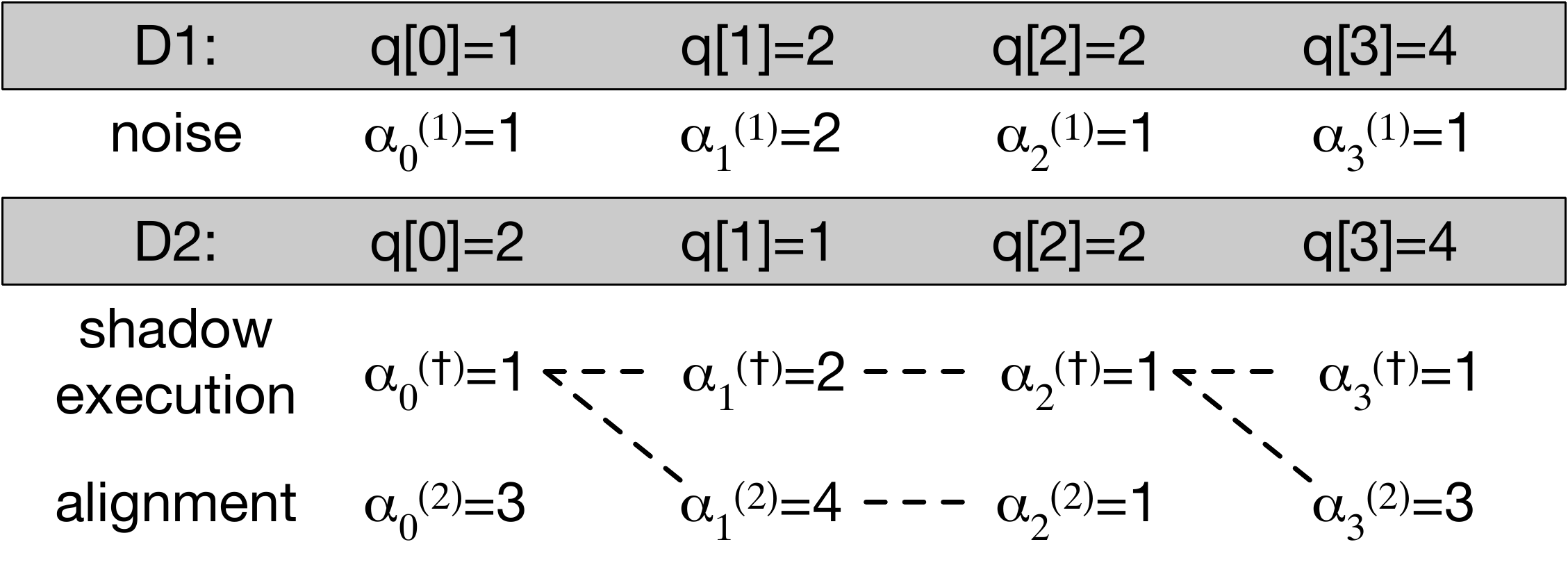}
\vspace{-2ex}
\caption{Selective alignment for Report Noisy Max}\label{fig:align}
\vspace{-2ex}
\end{figure}

We illustrate this construction in Figure~\ref{fig:align}. After seeing $q[1]$
on $D_1$ (Case 1), the construction uses noise in the shadow execution for
previous query answers, and uses $\alpha_1^{(1)}+2=4$ as the noise for $q[1]$ on
$D_2$.  After seeing $q[2]$ on $D_1$ (Case 2), the construction reuses
alignments constructed previously, and use $\alpha^{(1)}_2=1$ as the noise for
$q[2]$. When $q[3]$ comes, the previous alignment is abandoned; the shadow
execution is used for $q[0]$ to $q[2]$. It is easy to check that this
construction is correct for any subset of query answers seen so far, since the
resulting alignment is exactly the same as the informal proof above, when the
index of maximum value is known.

\paragraph{Randomness alignment with shadow execution}

To incorporate the informal argument above to a programming language, we
propose \lang. We
illustrate the key components of \lang in this section, as shown in
Figure~\ref{alg:transformed_noisymax}, and detail all components in the rest of
this paper.  

Similar to LightDP~\cite{lightdp}, \lang embeds randomness alignments into
types. In particular, each \emph{numerical variable} has a type in the form of
$\tyreal_{\pair {\first{\mathbb{n}}} {\second{\mathbb{n}}}}$, where
$\first{\mathbb{n}}$ and $\second{\mathbb{n}}$ represent the ``difference'' of
its value in the aligned and shadow execution respectively. In
Figure~\ref{alg:transformed_noisymax}, non-private variables, such as
$\epsilon, size$, are annotated with distance $0$. For private variables, the
difference could be a constant or an expression. For example, the type of $q$
along with the precondition specifies the adjacency relation: each query
answer's difference is specified by $*$, which is desugared to a special
variable $\first{\distance{q}}[i]$ (details discussed in
Section~\ref{sec:typing}).  The precondition in
Figure~\ref{alg:transformed_noisymax} specifies that the difference of each
query answer is bounded by 1 (i.e., query answers have sensitivity of 1).

\lang reasons about the aligned and shadow executions in isolation,
with the exception of sampling commands. A sampling command (e.g., line 3 in
Figure~\ref{alg:transformed_noisymax}) constructs the aligned execution by
either using values from the aligned execution so far (symbol $\alignd$), or
switching to values from the shadow execution (symbol $\shadow$). The
construction may depend on program state: in Figure~\ref{alg:transformed_noisymax}, we switch to shadow
values iff $q[i]+\eta$ is the max on $D_1$. A sampling command also specifies the
alignment for the generated random noise.

With function specification and annotations for sampling commands, the type
system of \lang automatically checks the source code. If successful, it
generates a non-probabilistic program (as shown at the bottom of
Figure~\ref{alg:transformed_noisymax}) with a distinguished variable
$\vpriv{\priv}$. The soundness of the type system ensures the following
property: if $\vpriv{\priv}$ is bounded by some constant $\priv$ in the
transformed program, then the original program being verified is
$\priv$-private. 

\paragraph{Benefits}

Compared with previous language-based proofs of Report Noisy
Max~\cite{Barthe16,Aws:synthesis} (both are based on the pointwise lifting
argument), \lang enjoys a unique benefit: the transformed code can be verified
based on \emph{standard} program semantics. Hence, the transformed
(non-probabilistic) program can be further analyzed by existing program
verifiers and other tools.  For example, the transformed program in
Figure~\ref{alg:transformed_noisymax} is verified with an off-the-shelf tool
CPAChecker\cite{beyer2011cpachecker} \emph{without any extra annotation} within seconds. Although not
explored in this paper, the transformed program can also be analyzed by
symbolic executors to identify counterexamples when the original program is
incorrect. We note that doing so will be more challenging in a customized
logic.

%% file: lang.tex
In this section, we present the syntax and semantics of \lang, a simple
imperative language for designing and verifying differentially private
algorithms.

\subsection{Syntax}
\label{sec:syntax}
The language syntax is given in Figure~\ref{fig:syntax}. Most parts of \lang is
standard; we introduce a few interesting features.

\input{syntax}

\paragraph{Non-probabilistic variables and expressions} 
\lang supports real numbers, booleans as well as standard operations on them.
We use $\nexpr$ and $\bexpr$ to represent numeric and boolean expressions
respectively. A ternary numeric expression ${\bexpr\mathbin{?}\nexpr_1:\nexpr_2}$ evaluates
to $\nexpr_1$ when the comparison evaluates to true, and $\nexpr_2$ otherwise.
Moreover, to model multiple queries to a database and produce multiple outputs
during that process, \lang supports lists: $e_1::e_2$ appends the element $e_1$
to a list $e_2$; $e_1[e_2]$ gets the $e_2$-th element in list $e_1$, assuming
$e_2$ is bound by the length of $e_1$.

\paragraph{Random variables and expressions} 
To model probabilistic computation, which is essential in differentially
private algorithms, \lang uses random variable $\eta\in \RVars$ to
store a sample drawn from a distribution.  Random variables are similar to
normal variables ($x\in \NVars$) except that they are the only ones who can get
random values from random expressions, via a sampling command $\eta:=g$. 

We follow the modular design of LightDP~\cite{lightdp}, where randomness
expressions can be added easily. In this paper, we only consider the most
interesting random expression, $\lapm r$. Semantically, $\eta:=\lapm r$ draws
one sample from the Laplace distribution, with mean zero and a scale factor
$r$, and assigns it to $\eta$. For verification purpose, a
sampling command also requires a few annotations, which we explain shortly.

\paragraph{Types}
Types in \lang have the form of $\basety_{\pair
{\first{\dexpr}} {\second{\dexpr}}}$, where $\basety$ is the base type, and $\first\dexpr$,
$\second{\dexpr}$ represent the alignments for the execution on adjacent
database and shadow execution respectively. Base type is standard: it includes
$\tyreal$ (numeric type), $\bool$ (Boolean), or a list of elements with type
$\tau$ ($\tylist~\tau$). 

Distance $\dexpr$ is the key for randomness alignment proof.
Intuitively, it relates two program executions so that the likelihood of seeing
each is bounded by some constant. Since only numerical values have numeric
distances, other data types (including $\bool$, $\tylist~\tau$ and
$\tau_1\rightarrow \tau_2$) are always associated with $\pair 0 0$, hence
omitted in the syntax. Note that this does not rule out numeric distances in
nested types. For example, $(\tylist~\tyreal_{\pair 1 1})$ stores numbers that
differ by exactly one in both aligned and shadow executions.

Distance $\dexpr$ can either be a numeric expression ($\nexpr$) in the language
or $*$. A variable $x$ with type $\tyreal_{\pair * *}$
is desugared as $x:\Sigma_{(\pair
{\first{\distance{x}}:\tyreal_{\pair 0 0}}{\second{\distance{x}}:\tyreal_{\pair 0 0}})}~\tyreal_{\pair
{\first{\distance{x}}}{\second{\distance{x}}}}$, where $\first{\distance{x}}$,
$\second{\distance{x}}$ are distinguished variables invisible in the source
code; hiding those variables in a $\Sigma$-type simplifies the type system
(Section~\ref{sec:typing}).

The star type is useful for two reasons. First, it specifies the sensitivity of
query answers in a precise way. Consider the parameter $q$ in
Figure~\ref{alg:transformed_noisymax} with type $\tylist~\tyreal_{\pair *
*}$, along with the precondition $\forall i \ge 0.~-1\leq \first{\distance{\cod{q}}}\cod{[i]} \leq 1$. This notation makes the assumption of the Report Noisy Max
algorithm explicit: each query answer differs by at most 1. Second, star type
serves as the last resort when the distance of a variable cannot be tracked
precisely by a static type system. For example, whenever \lang merges two
different distances (e.g., $3$ and $4$) of $x$ from two branches, the result
distance is $*$; the type system instruments the source code to maintain
the correct values of $\first{\distance{x}},\second{\distance{x}}$ (Section~\ref{sec:typing}).

\paragraph{Sampling with selectors}
Each sampling instruction is attached with a few annotations for proving
differential privacy, in the form of $(\eta :=
\lapm{\real},\select,\nexpr_{\eta})$.  Note that just like types, the
annotations $\select$, $\nexpr_{\eta}$ have no effects on the program
semantics; they only show up in verification. Intuitively, a selector $\select$
picks a version ($k\in \{\alignd, \shadow\}$) for all program variables
(including the previously sampled variables) at the sampling instruction, as
well as constructs randomness alignment for $\eta$, specified by
$\nexpr_{\eta}$ (note that the distance cannot be $*$ by syntactical restriction
here). By definition, both $\select$ and $\nexpr_{\eta}$ may depend on the
program state when the sampling happens. 

Return to the running example in Figure~\ref{alg:transformed_noisymax}. As
illustrated in Figure~\ref{fig:align}, the selective alignment is to 
\begin{itemize}
\item use shadow variables and align the new sample by $2$ whenever
a new max is encountered,

\item use aligned variables and the same sample otherwise.
\end{itemize}

Hence, the sampling command in Figure~\ref{alg:transformed_noisymax} is
annotated as $(\eta := \lapm{(2/\epsilon)},\Omega\mathbin{?}\shadow:\alignd,\Omega\mathbin{?}2:0)$, where $\Omega$
is $\cod{q[i]+}\eta>\cod{bq}\lor \cod{i=0}$, the condition when a new max is found.

\subsection{Semantics}
\label{sec:semantics}

As standard, the denotational semantics of the probabilistic language is defined as a
mapping from initial memory to a distribution on (possible) final outputs.
Formally, let $\Store$ be a set of memory states where each $\store \in \Store$
maps all (normal and random) variables ($\NVars \cup \RVars$) to
their values.

The semantics of an expression $e$ of base type $\basety$ is interpreted as a
function $\trans{e}: \Store \rightarrow  \trans{\basety}$, where
$\trans{\basety}$ represents the set of values belonging to the base type
$\basety$. We omit expression semantics since it is standard.
A random expression $g$ is interpreted as a distribution on real values. Hence,
$\trans{g}:\dist(\trans{\tyreal})$. Moreover, a command $c$ is interpreted as a
function $\trans{c}:\Store \rightarrow  \dist(\Store)$. For brevity, we write
$\evalexpr{e}{\store}$ and $\evalexpr{c}{\store}$ instead of
$\trans{e}(\store)$ and $\trans{c}(\store)$ hereafter.
Finally, all programs have the form $(c;\outcmd{e})$ where $c$ contains no 
return statement. A program is interpreted as a function $\store
\rightarrow  \dist{\trans{\basety}}$ where $\basety$ is the return type (of
$e$).

The semantics of commands is available in the \appendixref; the semantics
directly follows a standard semantics in \cite{Kozen81}.

\section{\lang: Type System}
\label{sec:typing}

\lang is equipped with a flow-sensitive type system. If successful, it
generates a transformed program with needed assertions to make the original
program differentially private. The transformed program is simple enough to be
verified by off-the-shelf program verifiers.

\subsection{Notations}
\label{sec:notations}
\input{typingrules}
We denote by $\Gamma$ the typing environment which tracks the type of each
variable in a flow-sensitive way (i.e., the type of each variable at each
program point is traced separately). All typing rules are formalized in
Figure~\ref{fig:typingrules}. Typing rules share a common global invariant
$\pre$, such as the sensitivity assumption annotated in the source code
(e.g., the precondition in Figure~\ref{alg:transformed_noisymax}). We also
write $\Gamma(x)=\pair{\first{\dexpr}}{\second{\dexpr}}$ for $\exists~\basety.~\Gamma(x)=\basety_{\pair{\first{\dexpr}}{\second{\dexpr}}}$ when
the base type $\basety$ is irrelevant.

\subsection{Expressions}
\label{sec:typing_expressions}
Expression rules have the form of $\Gamma\proves e:\tau$, which means that
expression $e$ has type $\tau$ under the environment $\Gamma$. Most rules are
straightforward: they compute the distance for aligned and shadow executions
separately.  Rule~\ruleref{T-OTimes} makes a conservative approach for
nonlinear computations, following LightDP~\cite{lightdp}.
Rule~\ruleref{T-VAR} desugars star types when needed.
The most interesting rule is~\ruleref{T-ODot}, which generates the following constraint:
\[\pre \Rightarrow (e_1\odot e_2 \Leftrightarrow (e_1\!+\!\nexpr_1)\odot (e_2\!+\!\nexpr_3) \land (e_1\!+\nexpr_2) \odot (e_2\!+\!\nexpr_4))\]

This constraint states that the boolean value of $e_1\odot e_2$ is identical in
both aligned and shadow executions. If the constraint is discharged by an external
solver (our type system uses Z3~\cite{z3}), we are assured that $e_1\odot
e_2$ has distances $\pair 0 0$.

\subsection{Commands}
\label{sec:typing_commands}

The flow-sensitive type system tracks and checks the distances of aligned
and shadow executions at each program point. Typing rules for commands have the
form of
\[\pc \proves \Gamma~\{ c\transform c'\}~\Gamma'\]
meaning that starting from the previous typing environment $\Gamma$, the new
typing environment is $\Gamma'$ after $c$. We will discuss the other components
$\pc$ and $c'$ shortly.

\subsubsection{Aligned Variables}
\label{sec:alignvars}

The type system infers and checks the distances of both aligned and shadow
variables. Since most rules treat them in the same way, we first explain the
rules with respect to aligned variables only, then we discuss shadow variables in
Section~\ref{sec:shadowvars}. To simplify notation, we write $\Gamma$ instead
of $\first{\Gamma}$ for now since only aligned variables are
discussed.

\paragraph{Flow-Sensitivity}
In each typing rule $\pc \proves \Gamma~\{c\transform c'\}~\Gamma'$, an
important invariant is that if $c$ runs on two memories that are aligned by
$\Gamma$, then the final memories are aligned by $\Gamma'$.

Consider the assignment rule~\ruleref{T-Asgn}. This rule computes the distance
of $e$'s value, $\first{\nexpr}$, and updates the distance of $x$'s value after
assignment to $\first{\nexpr}$. 

More interesting are rules~\ruleref{T-If}
and~\ruleref{T-While}. In~\ruleref{T-If}, we compute the typing environments
after executing $c_1$ and $c_2$ as $\Gamma_1$ and $\Gamma_2$ respectively.
Since each branch may update $x$'s distance in arbitrary way, $\Gamma_1(x)$ and
$\Gamma_2(x)$ may differ. We note that numeric expressions and $*$ type
naturally form a two level lattice, where $*$ is higher than any $\nexpr$.
Hence, we use the following rule to merge two distances $\dexpr_1$ and
$\dexpr_2$:
\[
\dexpr_1\join \dexpr_2 \defn
\begin{cases}
\dexpr_1 & \text{if } \dexpr_1 = \dexpr_2  \cr
* & \text{otherwise } \cr
\end{cases}
\]
For example, $(3\join 4=*)$, $(x+y\join x+y=x+y)$, $(x\join 3=*)$. Hence, \ruleref{T-If} ends with $\Gamma_1\join \Gamma_2$,
defined as $\lambda x.~\Gamma_1(x) \join \Gamma_2(x)$. 

As an optimization, we also use branch conditions to simplify distances.
Consider our running example (Figure~\ref{alg:transformed_noisymax}): at
Line~\ref{line:noisymax_branch}, $\eta$ has (aligned) distance
$\Omega\mathbin{?}2:0$, where $\Omega$ is the branch condition. Its distance is
simplified to $2$ in the $\true$ branch and $0$ in the $\false$ branch.

Rule \ruleref{T-While} is similar, except that it requires a fixed
point $\Gamma_f$ such that $\pc \proves \Gamma\join\Gamma_f\ \{c\}\ \Gamma_f$. In
fact, this rule is deterministic since we can construct the fixed point as
follows (the construction is similar to the one in~\cite{Hunt:flowsensitive}):
\[\pc \proves \Gamma_i'~\{c\transform c_i'\}~\Gamma_i''\text{ for all } 0\leq
i\leq n\]
where $\Gamma_0'=\Gamma, \Gamma_{i+1}'=\Gamma_i''\join \Gamma,
\Gamma_{n+1}'=\Gamma_n'$.

It is easy to check that $\Gamma_{n}'=\Gamma_{n+1}'=\Gamma_{n}''\join \Gamma$
and $\pc' \proves \Gamma_n'\ \{c\transform c_i'\}\  \Gamma_n''$ by construction.
Hence, $\Gamma_n''$ is a fixed point: $\pc \proves \Gamma \join \Gamma_n''
\{c\transform c_i'\}\ \Gamma_n''$. Moreover, the computation above always
terminates since all typing rules are monotonic on
typing environments\footnote{That is, $\forall \pc, c, \Gamma_1, \Gamma_2,
\Gamma_1', \Gamma_2', c_1, c_2.~\pc \proves \Gamma_i \{c\transform c_i'\}
\Gamma_i'\ i\in\{1,2\} \land \Gamma_1\sqsubseteq \Gamma_2 \implies
\Gamma_1'\sqsubseteq \Gamma_2'$.}
and the lattice has a height of 2.

\paragraph{Maintaining dynamically tracked distances}

Each typing rule $\pc \proves \Gamma~\{c\transform c'\}~\Gamma'$ also sets the
value of $\first{\distance{x}}$ to maintain distance dynamically
whenever $\Gamma'(x)=*$. This is achieved by the instrumented commands in $c'$.

None of rules \ruleref{T-Skip, T-Asgn, T-Seq, T-Ret} generate $*$ type, hence
they do not need any instrumentation. The merge operation in
rule~\ruleref{T-If} generates $*$ type when $\Gamma_1(x)\not=\Gamma_2(x)$. In
this case, we use the auxiliary
instrumentation rule in the form of $\Gamma_1, \Gamma_2,
\pc \Rrightarrow c'$, assuming $\Gamma_1\sqsubseteq \Gamma_2$. In particular,
for each variable $x$ whose distance is ``upgraded'' to $*$, the rule sets $\first{\distance{x}}$
to the distance previously tracked by the type system ($\Gamma_1(x)$). Moreover, the
instrumentation commands $c_1'',c_2''$ are inserted under their corresponding
branches.

Consider the following example:
\begin{lstlisting}[frame=none, numbers=none]
if (x > 1) x := y; else y := 1;
\end{lstlisting}
staring with $\Gamma_0:\{x:1, y: 0\}$. In the $\true$ branch,
rule~\ruleref{T-Asgn} updates $x$ to the distance of $y$, resulting
$\Gamma_1:\{x:0, y: 0\}$.  Similarly, we get $\Gamma_2:\{x:1, y: 0\}$ in the
$\false$ branch. Moreover, when we merge the typing environments $\Gamma_1$ and $\Gamma_2$ at the end of branch, the typing environment becomes $\Gamma_3 = \Gamma_1 \join \Gamma_2=\{x: *, y: 0\}$.
Since $\Gamma_1(x) \neq \Gamma_2(x)$, instrumentation rule is also applied,
which instruments $\first{\distance{\cod{x}}}\texttt{ := 0}$ after \texttt{x := y}
and $\first{\distance{\cod{x}}}\texttt{ := 1}$ after \texttt{y := 1}.

Rule~\ruleref{T-While} may also generate $*$ types. Following the same process in rule~\ruleref{T-If}, it also uses the
instrumentation rule to update corresponding dynamically tracked distance
variables. The instrumentation command $c_s$ is inserted before loop and $c''$ after the commands in the loop body.

\paragraph{Well-Formedness}

Whenever an assignment $x:=e$ is executed, no variable's distance should depend
on $x$. To see why, consider $x := 2$ with initial $\first{\Gamma}(y)=x$ and
$m(x)=1$. Since this assignment does not modify the value of $y$, the aligned
value of $y$ (i.e., $y+\first\Gamma(y)$) should not change. However,
$\first{\Gamma}(y)$ changes from 1 to 2 after the assignment.

To avoid this issue, we check the following condition for each assignment
$x:=e$: $\forall y \in \Vars.~x \notin \Vars (\Gamma(y))$. In case that the
check fails for some $y$, we promote its distance to $*$, and use the auxiliary
instrumentation $\Rrightarrow$ to set $\first{\distance{y}}$ properly. Hence,
\emph{well-formedness} is guaranteed: no variable's distance depends on $x$ when
$x$ is updated.

\paragraph{Aligned branches}

For differential privacy, we require the aligned execution to follow the same
branch as the original execution. Due to dynamically tracked distances,
statically checking that in a type system could be imprecise. Hence, we use
assertions in rules~\ruleref{T-If} and~\ruleref{T-While} to ensure the aligned
execution does not diverge. In those rules, $\alignexec{e, \Gamma}$ simply
computes the value of $e$ in the aligned execution; its full definition is in the \appendixref.  

\subsubsection{Shadow Variables}
\label{sec:shadowvars}

In most typing rules, shadow variables are handled in the same way as aligned
ones, which is discussed above. The key difference is that the type system
allows the shadow execution to take a different branch from the original
execution. 

The extra permissiveness is the key ingredient of verifying
algorithms such as Report Noisy Max.
To see why, consider the example in Figure~\ref{fig:align}, where the shadow execution
runs on $D_2$ with same random noise as from the execution on $D_1$. Upon the
second query, the shadow execution does not update max, since its noisy value
$3$ is the same as the previous max; however, execution on $D_1$ will update
max, since the noisy query value of $4$ is greater than the previous max of
$2$.  

To capture the potential divergence of shadow execution, each typing rule is
associated with a program counter $\pc$ with two possible values $\bot$ and
$\top$ (introducing program counters in a type system is common in
information flow control to track implicit
flows~\cite{sabelfeld2003language}). Here, $\top$ (resp. $\bot$) means that
the shadow execution might take a different branch (resp. must take the same
branch) as the original execution.

When $\pc=\bot$, the shadow execution is checked in the same way as aligned
execution.  When $\pc=\top$, the shadow distances are updated (as done in
Rule~\ruleref{T-Asgn}) so that $x+\second{\distance{x}}$ remains the same. 
The new value from the shadow execution will be maintained by the type system
\emph{when $\pc$ transits from $\bot$ to $\top$} by code instrumentation for
sub-commands in~\ruleref{T-If} and~\ruleref{T-While}, as we show next.

Take a branch $(\ifcmd{e}{c_1}{c_2})$ for example. The transition happens when
$\pc=\bot\land \pc'=\top$.
In this case, we construct a shadow execution of $c$ by an auxiliary function
$\shadowexec{c,\Gamma}$. The shadow execution essentially replaces each
variable $x$ with their correspondence (i.e., $x + \second{\distance{x}}$), as is
standard in self-composition~\cite{barthe2004, terauchi2005}. The
only difference is that
$\shadowexec{c,\Gamma}$ is not applicable to sampling commands, since we are
unable to align the sample variables when different amount of samples are
taken.
The full definition of $\shadowexec{c,\Gamma}$ is available in the \appendixref. 
 Rule~\ruleref{T-While} is very
similar in its way of handling shadow variables.

\subsubsection{Sampling Command}
\label{sec:sampling}

Rule~\ruleref{T-Laplace} checks the only probabilistic command
$\eta:=\lapm{\real}\mathbin{,}\select\mathbin{,}\nexpr_{\eta}$ in \lang. Here,
the selector $\select$ and numeric distance $\nexpr_{\eta}$ are annotations
provided by a programmer to aid type checking. For the sample $\eta$, the
aligned distance is specified by $\nexpr_{\eta}$ and the shadow distance is
always 0 (since by definition, shadow execution use the same sample as the original program).
Hence, the type of $\eta$ becomes $\tyreal_{\pair {\nexpr_{\eta}} 0}$.

Moreover, the selector constructs the aligned execution from either the aligned
($\alignd$) or shadow ($\shadow$) execution. Since the selector may depend on a
condition $e$, we use the selector function $\select(\pair {e_1} {e_2})$ in
Figure~\ref{fig:typingrules} to do so.

Rule~\ruleref{T-Laplace} also checks that each $\eta$ is generated in an
injective way: the same aligned value of $\eta$ implies the same value of
$\eta$ in the original execution.

Consider the sampling command in Figure~\ref{alg:transformed_noisymax}. The
typing environments before and after the command is shown below (we omit
unrelated  parts for brevity): 
\begin{lstlisting}[frame=none, numbers=none]
$\{\cod{bq}: \pair * *, \cdots\}$
$\eta$ :=  $\lapm(2/\priv)\mathbin{,}\Omega\mathbin{?}\shadow:\alignd\mathbin{\mathbin{,}}\Omega\mathbin{?}2:0\text{;}$
$\{\cod{bq}: \pair {\Omega\mathbin{?}{\second{\distance{\cod{bq}}}}:\first{\distance{\cod{bq}}}}{{\second{\distance{\cod{bq}}}}}\mathbin{,}\mathtt{\eta}: \pair {\Omega\mathbin{?}2:0} 0\mathbin{,}\cdots\}$
\end{lstlisting}

In this example, $\select$ is $\Omega\mathbin{?}\shadow:\alignd$. So the aligned
distance of variable $\cod{bq}$ will be 
$\Omega\mathbin{?}{\second{\distance{\cod{bq}}}}:\first{\distance{\cod{bq}}}$, the shadow distance of
variable $\cod{bq}$ is still $\second{\distance{\cod{bq}}}$. The aligned distance of $\eta$
is $\pair {\Omega\mathbin{?}2:0} 0$, where the aligned part is specified in the annotation.

\subsection{Target Language}
\label{sec:target_language}

One goal of \lang is to enable verification of $\priv$-differential privacy
using off-the-shelf verification tools. In the transformed code so far, we
assumed $\cod{assert}$ commands to verify that certain condition holds. The
only remaining challenging feature is the
sampling commands, which requires probabilistic reasoning. Motivated by
LightDP~\cite{lightdp}, we note that for $\priv$-differential privacy, we are only
concerned with the maximum privacy cost, not its likelihood. Hence, in the final
step, we simply replace the sampling command with a non-deterministic command
$\havoc \eta$, which semantically sets the variable $\eta$ to an arbitrary value upon
execution, as shown in Figure~\ref{fig:totarget}.

Note that a distinguished variable $\vpriv{\priv}$ is added by the type
system to explicitly track the privacy cost of the original program. For
Laplace distribution, aligning $\eta$ by the distance of $\nexpr_{\eta}$ is
associated with a privacy cost of $|\nexpr_{\eta}|/r$. The reason is that the
ratio of any two points that are $|\nexpr_{\eta}|$ apart in the Laplace
distribution with scaling factor $r$ is bounded by $\exp(|\nexpr_{\eta}|/r)$.
Since the shadow execution uses the same sample, it has no privacy cost. This
very fact allows us to \emph{reset} privacy cost when the shadow execution is 
used (i.e., $\select$ selects $\shadow$): the rule sets privacy cost
to $0+|\nexpr_{\eta}|/r$ in this case.

In Figure~\ref{alg:transformed_noisymax}, $\vpriv{\priv}$ is set to
$\Omega\mathbin{?}0:\vpriv{\priv}+\Omega\mathbin{?}\epsilon:0$ which is the same as
$\Omega\mathbin{?}\epsilon:\vpriv{\priv}$. Intuitively, that implies that the privacy
cost of the entire algorithm is either $\epsilon$ (when a new max is found) or
the same as the previous value of $\vpriv{\priv}$.

The type system guarantees the following important property: if the original program
type checks and the privacy cost $\vpriv{\priv}$ in the target language is
bounded by some constant $\priv$ in all possible executions of the program,
then the original program satisfies $\priv$-differential privacy. We will
provide a soundness proof in the next section. Consider the
running example in Figure~\ref{alg:transformed_noisymax}. The transformed
program in the target language is shown at the bottom. With a model checking
tool CPAChecker~\cite{beyer2011cpachecker}, we verified that $\vpriv{\priv}\leq \epsilon$ in the
transformed program within 2 seconds (Section~\ref{sec:evaluation}). Hence, the
Report Noisy Max algorithm is verified to be $\epsilon$-differentially private.

\begin{figure}
\setstretch{0.4}
\begin{mathpar}
\inferrule{ }{\eta := \lapm~\real;\select,\nexpr_{\eta} \rightrightarrows \havoc~\eta;\vpriv{\priv} :=\select(\pair {\vpriv{\priv}} 0)+|\nexpr_{\eta}|/r;}
\and
\inferrule{ }{
c \rightrightarrows c \text{, if } c \text{ is not a sampling command}
}
\end{mathpar}
\caption{Transformation rules to the target language. Probabilistic
commands are reduced to non-deterministic ones.}
\label{fig:totarget}
\end{figure}

%% file: syntax.tex
\begin{figure}
\setstretch{0.9}
$
\begin{array}{lccl}
\text{Reals} & r &\in &\mathbb{R} \\
\text{Normal Vars} & x  &\in &\NVars \\
\text{Random Vars} & \eta  &\in &\RVars \\
\text{Linear Ops} & \oplus &::= &+ \mid - \\
\text{Other Ops} & \otimes &::= &\times \mid / \\
\text{Comparators} & \odot &::= &< \mid > \mid = \mid \leq \mid \geq \\
\text{Bool Exprs} & \bexpr &::=\; &\true \mid \false \mid x \mid \neg \bexpr \mid \nexpr_1 \odot \nexpr_2 \\
\text{Num Exprs} & \nexpr &::=\; &\real \mid x \mid \eta  \mid \nexpr_1\oplus \nexpr_2 \mid \\ & & & \nexpr_1\otimes \nexpr_2 \mid \bexpr?\ \nexpr_1 : \nexpr_2 \\
\text{Expressions} & e &::=\; &\nexpr \mid \bexpr \mid e_1::e_2 \mid e_1[e_2] \\
\text{Commands} & c &::=\; &\skipcmd \mid x := e \mid \eta  := g \mid c_1;c_2 \mid  \\ & & & \outcmd{e} \mid \whilecmd{e}{(c)} \mid \\ & & & \ifcmd{e}{(c_1)}{(c_2)}\\
\text{Distances}& \dexpr &::= &\nexpr \mid *\\
\text{Types} & \tau &::=\; &\tyreal_{\pair{\first{\dexpr}}{\second{\dexpr}}} \mid \bool \mid \tylist~\tau \\ %
\text{Var Versions} & k &\in &\{\alignd, \shadow\} \\
\text{Selectors} & \select &::=\; &e\ ?\ \select_1:\select_2 \mid k \\
\text{Rand Exps}   & g &::=\; &\lapm{\real}\mathbin{,}\select\mathbin{,}\nexpr_{\eta} %
\end{array}
$
\vspace{-2ex}
\caption{\lang: language syntax.}
\vspace{-2ex}
\label{fig:syntax}
\end{figure}

%% file: typingrules.tex
\begin{figure*}
\raggedright
\setstretch{0.9}
\framebox{\textbf{Typing rules for expressions}}
\vspace{-1em}
\begin{mathpar}
\inferrule*[right=(T-Num)]{ }{ \Gamma  \proves \real : \tyreal_{\pair 0 0}}
\quad
\inferrule*[right=(T-Boolean)]{ }{ \Gamma \proves b : \bool}
\quad
\inferrule*[right=(T-Var)]{ \Gamma(x) = \basety_{\pair{\first{\dexpr}}{\second{\dexpr}}} 
\quad
\third{\nexpr} = 
\begin{cases}
\third{\distance{x}} & \text{if } \third{\dexpr} = * \cr
\third{\dexpr} & \text{otherwise}
\end{cases}
\quad
\star\in\set{\circ, \dagger}
}{\Gamma \proves x : \basety_{\pair{\first{\nexpr}}{\second{\nexpr}}}}
\and
\inferrule*[right=(T-OPlus)]{\Gamma \proves e_1: \tyreal_{\pair {\nexpr_1} {\nexpr_2} } \quad \Gamma\proves e_2 : \tyreal_{\pair {\nexpr_3} {\nexpr_4} }}{ \Gamma \proves e_1\oplus e_2 : \tyreal_{\pair {\nexpr_1 \oplus \nexpr_3} {\nexpr_2 \oplus \nexpr_4 } }}
\and
\inferrule*[right=(T-OTimes)]{\Gamma\proves e_1: \tyreal_{\pair 0 0} \quad \Gamma \proves e_2 : \tyreal_{\pair 0 0}}{\Gamma\proves e_1\otimes e_2 : \tyreal_{\pair 0 0}}
\and
\inferrule*[right=(T-ODot)]{\inferrule*{}{\Gamma \proves e_1: \tyreal_{\pair {\nexpr_1} {\nexpr_2} } \\\\ \Gamma \proves e_2 : \tyreal_{\pair {\nexpr_3} {\nexpr_4} }} \qquad
\inferrule*{}{\pre \Rightarrow 
 (e_1\odot e_2 \Leftrightarrow (e_1\!+\!\nexpr_1)\odot (e_2\!+\!\nexpr_3)) \\\\\hspace{4ex} \land
(e_1\odot e_2 \Leftrightarrow (e_1\!+\!\nexpr_2)\odot (e_2\!+\!\nexpr_4))
}} { \Gamma \proves e_1\odot e_2 : \bool}
\and
\inferrule*[right=(T-Ternary)]{\Gamma \proves e_1:\bool \quad \Gamma \proves e_2:\tau \quad \Gamma \proves e_3:\tau}{\Gamma \proves e_1\mathbin{?}e_2:e_3:\tau}
\and
\inferrule*[right=(T-Neg)]{\Gamma \proves e: \bool}{ \Gamma \proves \neg e : \bool}
\and
\inferrule*[right=(T-Cons)]{\Gamma \proves e_1:\tau \quad \Gamma\proves e_2:\tylist~\tau} {\Gamma\proves e_1::e_2 : \tylist~\tau}
\and
\inferrule*[right=(T-Index)]{\Gamma \proves e_1:\tylist~\tau\quad \Gamma\proves e_2:\tyreal_{\pair 0 0}}{\Gamma \proves e_1[e_2]:\tau}
\end{mathpar}
\vspace{0.5em}
\framebox{\textbf{Typing rules for commands}}
\vspace{-1em}
\begin{mathpar}
\inferrule*[right=(T-Skip)]{ }{\pc \proves \Gamma~\{\skipcmd  \transform \skipcmd\}~\Gamma}
\quad
\inferrule*[right=(T-Asgn)]{ 
\Gamma \proves e: \basety_{\pair {\first{\nexpr}}{\second{\nexpr}}} \quad
\configtwo {\Gamma'} {\second{c}}=
\begin{cases}
\configtwo {\Gamma[x\mapsto \basety_{\pair{\first{\nexpr}}{\second{\nexpr}}}]} \skipcmd, \ \text{ if } \pc=\bot \cr 
\annotation{
\configtwo {\Gamma[x\mapsto \basety_{\pair{\first{\nexpr}}{*}}]} 
{\second{\distance{x}} := x + \second{\nexpr} - e}, \ \text{ else } }
\end{cases}
}
{\pc \proves \Gamma~\{x := e \transform \second{c}; x := e\}~\Gamma' }
\and
\inferrule*[right=(T-Seq)]{
\inferrule{\pc \proves \Gamma~\{c_1 \transform c_1'\}~\Gamma_1\\
\pc \proves \Gamma_1~\{c_2 \transform c_2'\}~\Gamma_2}{}
}{\pc \proves \Gamma~\{c_1;c_2 \transform c_1';c_2'\}~\Gamma_2}
\and
\inferrule*[right=(T-Return)]{\Gamma \proves e:\tyreal_{\pair 0 \dexpr } {\text\quad or\quad } \Gamma \proves e:\bool}
{\pc \proves \Gamma~\{\outcmd{e} \transform %
\outcmd{e}\}~\Gamma}
\and
\inferrule*[right=(T-If)]{
\inferrule{
\pc' = \cod{updPC}(\pc,\Gamma,e) \\\\
\pc' \proves \Gamma~\{c_i \transform c_i'\}~\Gamma_i
}
{}
\quad 
\Gamma_i, \Gamma_1\join \Gamma_2,\pc' \Rrightarrow c_i''
\ 
i\in \{1,2\}
\
\quad
{\second{c}} \!=\!
\begin{cases}
\skipcmd , \quad \text{ if } (\pc=\top \lor \pc'=\bot) \cr
\annotation{{\shadowexec{\ifcmd{e}{c_1}{c_2},\Gamma_1\join\Gamma_2}},\text{ else}}
\end{cases}
}
{ \pc \proves \Gamma~\big\{\ifcmd{e}{c_1}{c_2} \transform 
\big(\ifcmd{e}{(\cod{assert}~(\alignexec{e, \Gamma});c_1';c_1'')}{(\cod{assert}~(\neg \alignexec{e, \Gamma});c_2'; c_2''})\big);\second{c}\big\}~\Gamma_1\join \Gamma_2
}
\and
\inferrule*[right=(T-While)]{
\inferrule{\pc' = \cod{updPC}(\pc,\Gamma,e) \\\\
\pc' \proves \Gamma \join \Gamma_f~\{c \transform c'\}~\Gamma_f
}{}
\quad
\inferrule{
\Gamma, \Gamma\join \Gamma_f, \pc' \Rrightarrow c_s \\\\
\Gamma_f, \Gamma\join \Gamma_f, \pc' \Rrightarrow c''}{}
\quad
\second{c} =
\begin{cases}
\skipcmd , \quad\text{ if } (\pc=\top \lor \pc'=\bot) \cr
\annotation{\shadowexec{\whilecmd{e}{c},\Gamma\join \Gamma_f}, \quad\text{ else}}
\end{cases}
}{\pc \proves \Gamma~\{\whilecmd{e}{c} \transform c_s; (\whilecmd{e}{(\cod{assert}~(\alignexec{e, \Gamma});c';c'')});\second{c}\}~\Gamma\join \Gamma_f}
\end{mathpar}
\vspace{0.6em}
\framebox{\textbf{Typing rules for random assignments}}
\begin{mathpar}
\inferrule*[right=(T-Laplace)]{
\pc = \bot
\quad
\Gamma' = \lambda x.~
\pair {\select (\pair {\first{\nexpr}} {\second{\nexpr}})} {\second{\nexpr}}
\text{ where } \Gamma\proves x:\basety_{\pair {\first{\nexpr}} {\second{\nexpr}}}
\quad 
\pre\Rightarrow (\subst{(\eta + \nexpr_\eta)}{\eta}{\eta_1} = \subst{(\eta + \nexpr_\eta)}{\eta}{\eta_2} \Rightarrow \eta_1=\eta_2)
}
{\pc \proves \Gamma~\{\eta := \lapm~\real;\select,\nexpr_\eta \transform \eta:=\lapm~\real;\select,\nexpr_{\eta}\}~\Gamma'[\eta\mapsto \tyreal_{\pair {\nexpr_{\eta}} 0}]}
\end{mathpar}
\vspace{0.6em}
\framebox{\textbf{Instrumentation rule}}
\vspace{0.1em}
\begin{mathpar}
\inferrule*{\Gamma_1 \sqsubseteq \Gamma_2 \quad 
\inferrule*{}
{
\first{c}=\{\first{\distance x}:=\nexpr \mid \Gamma_1(x)=\tyreal_{\pair {\nexpr} {\dexpr_1}} \land \Gamma_2(x)=\tyreal_{\pair * {\dexpr_2}}\} \\\\
\second{c}=\{\second{\distance x}:=\nexpr \mid \Gamma_1(x)=\tyreal_{\pair {\dexpr_1} {\nexpr} } \land \Gamma_2(x)=\tyreal_{\pair {\dexpr_2} *}\} 
}
\\
c'=
\begin{cases}
\first{c};\second{c} & \text{ if } \pc=\bot \cr
\first{c} & \text{ if } \pc=\top \cr
\end{cases}
}{\Gamma_1, \Gamma_2, \pc \Rrightarrow c'}
\end{mathpar}
\vspace{0.2em}
\framebox{\textbf{Select function}}
\begin{mathpar}
\vspace{0.2em}
\alignd(\pair {e_1} {e_2})=e_1
\and
\shadow(\pair {e_1} {e_2})=e_2
\and
(e\mathbin{?}\select_1:\select_2)(\pair {e_1} {e_2})=
e\mathbin{?}\select_1(\pair {e_1} {e_2}):\select_2(\pair {e_1} {e_2}) 
\end{mathpar}
\vspace{0.2em}
\framebox{\textbf{PC update function}}
\vspace{-0.5em}
\begin{mathpar}
\cod{updPC}(\pc, \Gamma, e)=
\begin{cases}
\bot \text{ , if } \pc=\bot\land \pre \Rightarrow (e \Leftrightarrow \shadowexec{e, \Gamma}) \\
\top \text{ , else}
\end{cases}
\end{mathpar}
\vspace{-3ex}

\caption{Typing rules and auxiliary rules. $\pre$ is an invariant that holds
throughout program execution. In most rules, shadow distances are handled in
the same way as aligned distances, with exceptions highlighted in gray boxes.}
\label{fig:typingrules}
\end{figure*}

%% file: soundness.tex
The type system performs a two-stage transformation: 
\[\pc \proves \Gamma_1\ \{c \transform c'\}\ \Gamma_2 \quad \text{ and }\quad c' \rightrightarrows c''\]

Here, both $c$ and $c'$ are probabilistic programs; the difference is that $c$
executes on the original memory without any distance tracking variables; $c'$
executes on the extended memory where distance tracking variables 
are visible. In the second stage, $c'$ is transformed to a
non-probabilistic program $c''$ where sampling instructions are replaced by
$\havoc{}$ and the privacy cost $\vpriv{\priv}$ is explicit. In this section,
we use $c,c',c''$ to represent the source, transformed, and target program
respectively.

Overall, the type system ensures $\priv$-differential privacy
(Theorem~\ref{thm:privacy}): if the value of $\vpriv{\priv}$ in $c''$ is always
bounded by a constant $\epsilon$, then $c$ is $\priv$-differentially private.
In this section, we formalize the key properties of our type system and prove
its soundness. Due to space constraints, the complete proofs are available in
the \appendixref.

\paragraph{Extended Memory}

Command $c'$ is different from $c$ since it maintains and uses distance tracking variables.
To close the gap, we first extend memory $m$ to include those
variables, denoted as  $\DVars=\cup_{x\in
\NVars}\set{\first{\distance{x}},\second{\distance{x}}}$ and introduce a distance environment
$\gamma:\DVars \rightarrow \mathbb{R}$.

\begin{definition}Let $\gamma: \DVars\rightarrow \mathbb{R}$. For any $m\in \Store$, there is an extension of $m$,written $\extmemfull{m}{\gamma}{}$, such that 
\[\extmemfull{m}{\gamma}{} (x) =\begin{cases}
m(x), &x\in \Vars\\
\gamma(x), &x\in \DVars
\end{cases}\]
\end{definition}

We use ${\Store}'$ to denote the set of extended memory states and $m_1'$,
$m_2'$ to refer to concrete extended memory states.
We note that although the programs $c$ and $c'$ are probabilistic, the extra
commands in $c'$ are deterministic. Hence, $c'$ preserves the semantics of $c$,
as formalized by the following Lemma.

\begin{lemma}[Consistency]\label{lem:consistency}
Suppose $\pc\proves \Gamma_1~\set{c \transform c'}~\Gamma_2$. Then for any initial and final memory $m_1$, $m_2$ such that $\evalexpr{c}{m_1}(m_2)\neq 0$, and any extension $m_1'$ of $m_1$, there is a unique extension $m_2'$ of $m_2$ such that 
\[\evalexpr{c'}{m_1'}(m_2') = \evalexpr{c}{m_1}(m_2)\]
\end{lemma}
\begin{proof}
By structural induction on $c$. The only interesting case is the
(probabilistic) sampling command, which does not modify distance tracking
variables.
\end{proof}

From now on, we will use $m_2'$ to denote the unique extension of $m_2$ satisfying the property above. 

\paragraph{$\Gamma$-Relation}

To formalize and prove the soundness property, we notice that a typing
environment $\Gamma$ along with distance environment $\gamma$
induces two binary relations on memories. We write
$\extmemfull{m_1}{\gamma}{}~\first{\Gamma}~m_2$ (resp.
$\extmemfull{m_1}{\gamma}{}~ \second{\Gamma}~m_2$) when $m_1,m_2$ are related
by $\first{\Gamma}$ (resp. $\second{\Gamma}$) and $\gamma$. Intuitively, the
initial $\gamma$ and $\Gamma$ (given by the function signature) specify the
adjacency relation, and the relation is maintained by the type system
throughout program execution. For example, the initial $\gamma$ and $\Gamma$ in
Figure~\ref{alg:transformed_noisymax} specifies that two executions of the
program is related if non-private variables $\epsilon, size$ are identical, and each query answer in
$q[i]$ differs by at most one.

To facilitate the proof, we simply write $m_1'~\Gamma~m_2$ where $m_1'$ is
an extended memory in the form of $\extmemfull{m_1}{\gamma}{}$.

\begin{definition}[$\Gamma$-Relations]
\label{def:relation} Two memories $m_1'$ (in the form of
$\extmemfull{\store_1}{\gamma}{}$) and ${\store}_2$ are related by
$\first{\Gamma}$, written $m_1'~\first{\Gamma}~{\store}_2$, if $\forall x\in
\Vars\cup \RVars$, we have
\begin{align*}
\store_2(x)=m_1'(x) + m_1'(\first{\dexpr}) &\text{ if } \Gamma \proves x : \tyreal_{\pair {\first{\dexpr}} {\second{\dexpr}}}
\end{align*}
We define the relation on non-numerical types and the $\second{\Gamma}$
relation in a similar way.
\end{definition} 

By the definition above, $\first{\Gamma}$ introduces a function from
${\Store}'$ to ${\Store}$. Hence, we use $\first{\Gamma}m_1'$ as the
unique $m_2$ such that $m_1'~ \first{\Gamma}~m_2$.  The $\second{\Gamma}$
counterparts are defined similarly.

\paragraph{Injectivity}

For alignment-based proofs, given any $\gamma$, both $\first{\Gamma}$ and
$\second{\Gamma}$ must be injective functions~\cite{lightdp}.
The injectivity of $\Gamma$ over the entire memory follows from the injectivity
of $\Gamma$ over the random noises $\eta\in \RVars$, which is checked as the
following requirement in Rule \ruleref{T-Laplace}: 
$$\pre \Rightarrow (\subst{(\eta + \nexpr_\eta)}{\eta}{\eta_1} = \subst{(\eta + \nexpr_\eta)}{\eta}{\eta_2} \Rightarrow \eta_1=\eta_2)$$
where all variables are universally quantified. Intuitively, this is true since
the non-determinism of the program is purely from that of $\eta\in \RVars$. 

\begin{lemma}[Injectivity]
\label{lem:injective}
Given $c,c',\pc, m',m_1', m_2', \Gamma_1, \Gamma_2$ such that $\pc \proves \cmdfull{\Gamma_1}{\Gamma_2}{c\transform c'}$, 
$\evalexpr{c'}{m'}{m_1'}\not=0 \land \evalexpr{c'}{m'}{m_2'}\not=0, \star\in \{\alignd, \shadow\}$, then we have
\begin{align*}
\Gamma_2^\star m_1' = \Gamma_2^\star m_2'\implies  m_1' = m_2' 
\end{align*}
\end{lemma}

\paragraph{Soundness}
The soundness theorem connects the ``privacy cost'' of the probabilistic
program to the distinguished variable $\vpriv{\epsilon}$ in the target program
$c''$. To formalize the connection,
we first extend memory one more time to include $\vpriv{\epsilon}$:

\begin{definition}
For any extended memory $m'$ and constant $\priv$, there is an extension of $m'$,
written $\extmemfull{m'}{\priv}{}$, so that 
\[\extmemfull{m'}{\priv}{}(\vpriv{\epsilon})=\priv, ~\text{ and }~ \extmemfull{m'}{\priv}{}(x) =m(x),~\forall x\in \dom(m').\]
\end{definition}

For a transformed program and a pair of initial and final memories $m_1'$ and $m_2'$, we identify a set of possible $\vpriv{\epsilon}$ values, so that in the corresponding executions of $c''$, the initial and final memories are extensions
of $m_1'$ and $m_2'$ respectively:
\begin{definition} Given $c'\rightrightarrows c''$, $m_1'$ and $m_2'$, the consistent costs of executing $c''$ w.r.t.  $m_1'$ and
$m_2'$, written $\restrict{c''}{m_1'}{m_2'}$, is defined as
\[\restrict{c''}{m_1'}{m_2'} \defn \{\priv \mid 
\extmemfull{m_2'}{\priv}{0} \in \evalexpr{c''}{\extmemfull{m_1'}{0}{0}} \}\]
\end{definition}

Since $(\restrict{c''}{m_1'}{m_2'})$ by definition is a set of values of
$\vpriv{\epsilon}$, we write $\max(\restrict{c''}{m_1'}{m_2'})$ for the maximum
cost. 

The next lemma
enables precise reasoning of privacy cost w.r.t. a pair of initial and final
memories:
\begin{lemma}[Pointwise Soundness]
\label{lem:pointwise}
Let $\pc, c,c',c'',\Gamma_1, \Gamma_2$ be such that $\pc \proves \cmdfull{\Gamma_1}{\Gamma_2}{c\transform c'}~\land~ c'\rightrightarrows c''$, then $\forall m_1',m_2'$:
\begin{itemize}
\item[(i)] the following holds:
\begin{equation}\label{inv:shadow}
\evalexpr{c'}{\store_1'}(\store_2') \leq \evalexpr{c}{\second{\Gamma_1}\store_1'}(\second{\Gamma_2}\store_2') \text{ when } \pc=\bot
\end{equation}
\item[(ii)] one of the following holds:
\begin{subequations}
    \begin{align}
        \evalexpr{c'}{m_1'}(m_2')&\leq\exp(\max(\restrict{c''}{m_1'}{m_2'}))\evalexpr{c}{\first{\Gamma_1}m_1'}(\first{\Gamma_2}m_2')\label{inv:aligned}\\
         \evalexpr{c'}{m_1'}(m_2')&\leq\exp(\max(\restrict{c''}{m_1'}{m_2'}))\evalexpr{c}{\second{\Gamma_1}m_1'}(\first{\Gamma_2}m_2')\label{inv:aligned2}
    \end{align}
\end{subequations}
\end{itemize}
\end{lemma}

The point-wise soundness lemma provides a precise privacy bound per initial and
final memory. However, differential privacy by definition
(Definition~\ref{def:diffpriv}) bounds the worst-case cost. To close the gap,
we define the worst-case cost of the transformed program.

\begin{definition}
For any program $c''$ in the target language, we say the execution cost of $c''$  is
\emph{bounded} by some constants $\priv$, written
$c''^{\preceq \priv}$, iff for any $m_1', m_2'$,
\[\extmemfull{m_2'}{\priv'}{\delta'}\in \evalexpr{c''}{\extmemfull{m_1'}{0}{0}} \Rightarrow \priv'\leq\priv\]
\end{definition}

Note that off-the-shelf tools can be used to verify that $c''^{\preceq \priv}$
holds for some $\priv$.

\begin{theorem}[Soundness]
\label{thm:soundness}
Given $c,c',c'', m_1',\Gamma_1, \Gamma_2,\priv$ such that $\bot\!\proves\! \Gamma_1\set{c\!\transform\! c'}\Gamma_2~\land~ c'\!\rightrightarrows\! c''~\land ~c''^{\preceq \priv}$, one of the following holds:
\vspace{-1em}
\begin{subequations}
    \begin{align}
        \max_{S\subseteq\Store'}(\evalexpr{c'}{m_1'}(S)-\exp(\priv)\evalexpr{c}{\first{\Gamma}_1 m_1'}(\first{\Gamma_2}S)) &\leq 0,\label{thm:soundnessa}\\
        \max_{S\subseteq\Store'}(\evalexpr{c'}{m_1'}(S)-\exp(\priv)\evalexpr{c}{\second{\Gamma_1} m_1'}(\first{\Gamma_2}S)) &\leq 0.\label{thm:soundnessb}
    \end{align}
\end{subequations}
\end{theorem}
\begin{proof}
By definition of $c''^{\preceq \priv}$, we have
${\max(\restrict{c''}{m_1'}{m_2'})} \leq\priv$ for all $m_2'\in S$. Thus, by
Lemma~\ref{lem:pointwise}, we have one of the two:  
\[\evalexpr{c'}{m_1'}(m_2')\leq\exp(\priv)\evalexpr{c}{\aligndstore{\Gamma_1}m_1'}(\aligndstore{\Gamma_2}m_2'), \quad\ \forall m_2'\in S,\]
\[ \evalexpr{c'}{m_1'}(m_2')\leq\exp(\priv)\evalexpr{c}{\shadowstore{\Gamma_1}m_1'}(\aligndstore{\Gamma_2}m_2'), \quad\ \forall m_2'\in S.\]
If the first inequality is true, then
\begin{align*}
   &\max_{S\subseteq\Store'}(\evalexpr{c'}{m_1'}(S)-\exp(\priv)\evalexpr{c}{\first{\Gamma_1}m_1'}(\first{\Gamma_2}S)) \\
 = &\max_{S\subseteq\Store'}\sum_{m_2'\in S} (\evalexpr{c'}{m_1'}(m_2')-\exp(\priv)\evalexpr{c}{\first{\Gamma_1}m_1'}(\first{\Gamma_2}m_2'))
 \leq 0
\end{align*}
and therefore \eqref{thm:soundnessa} holds.
Similarly, \eqref{thm:soundnessb} holds if the second inequality is true.
Note that the equality above holds due to the injective
assumption, which allows us to derive the set-based privacy from
the point-wise privacy (Lemma~\ref{lem:pointwise}).
\end{proof}

We now prove the main theorem on differential privacy:

\begin{theorem}[Privacy]
\label{thm:privacy}
Given $\Gamma_1,\Gamma_2,c,c',c'', e,\priv$ such that $\first{\Gamma_1}=\second{\Gamma_1} \land \bot \proves \Gamma_1\set{(c;\outcmd~e) \transform (c';\outcmd~e)}\Gamma_2 ~\land ~c' \rightrightarrows c''$, we have 
\[c''^{\preceq \priv} \Rightarrow c \text{ is } \priv\text{-differentially private}.\]

\end{theorem}

\begin{proof}
By the typing rule, we have $\bot \proves \Gamma_1\set{c\transform
c'}\Gamma_2$.  By the soundness theorem (Theorem~\ref{thm:soundness}) and the
fact that $\first{\Gamma}_1=\second{\Gamma}_2$, we have 
$\evalexpr{c'}{m_1'}(S)\leq\exp(\priv)\evalexpr{c}{\aligndstore{\Gamma_1}{m_1'}}(\aligndstore{\Gamma_2}{S})$.
For clarity, we stress that all sets are over
distinct elements (as we have assumed throughout this paper). 

By rule~\ruleref{T-Return}, $\Gamma_2 \proves e : \tyreal_{\pair{0}{\dexpr}}$ or $\Gamma_2 \proves e : \bool$. For any set of values $V\subseteq \trans{\basety}$, let $S_V'=\{m'\in \Store' \mid \evalexpr{e}{m'} \in V\}$ and $S_V=\{m\in \Store \mid \evalexpr{e}{m} \in V\}$, then we have $\amem{\Gamma_2}S_V'\subseteq S_V$:
\begin{align*}
m\in \amem{\Gamma_2}{S_V'} &\Rightarrow m = \amem{\Gamma_2}{m'} \text{ for some } m'\in S_V\\
&\Rightarrow \evalexpr{e}{m} = \evalexpr{e}{\amem{\Gamma_2}{m'}} = \evalexpr{e}{m'}\in V \\
&\Rightarrow m\in S_V.
\end{align*}
The equality in second implication is due to the zero distance when $\Gamma_2 \proves e : \tyreal_{\pair{0}{\nexpr}}$, and rule~\ruleref{T-ODot} when $\Gamma_2 \proves e : \bool$. We note that $\amem{\Gamma_2}S_V'\not= S_V$ in general since $\first{\Gamma_2}$ might not be a surjection. Let $P'=(c';\outcmd~e)$, then for any $\gamma$, we have
\begin{align*}
\evalexpr{P'}{\extmemfull{m_1}{\gamma}{}}(V) &=  \evalexpr{c'}{\extmemfull{m_1}{\gamma}{}}(S_V') \\
                     &\leq  \exp(\priv) \evalexpr{c'}{\amem{\Gamma_1}{\extmemfull{m_1}{\gamma}{}}}(\amem{\Gamma_2}S_V') \\
                     &\leq  \exp(\priv) \evalexpr{c'}{\amem{\Gamma_1}{\extmemfull{m_1}{\gamma}{}}}(S_V) \\
                     &=\exp(\priv) \evalexpr{P}{\amem{\Gamma_1}{\extmemfull{m_1}{\gamma}{}}}(V).
\end{align*}
Finally, due to Lemma~\ref{lem:consistency},
$\evalexpr{P}{m_1}(V)=\evalexpr{P'}{\extmemfull{m_1}{\gamma}{}}(V)$.
Therefore, by definition of privacy $c$ is $\priv$-differentially private.
\end{proof}

Note that the shallow distances are only useful for proofs; they are irrelevant
to the differential privacy property being obeyed by a program. Hence,
initially, we have $\first{\Gamma}_1=\second{\Gamma}_1$ (both describing the
adjacency requirement) in Theorem 2, as well as in all of the examples formally
verified by \lang.

%% file: evaluation.tex
\subsection{Implementation}\label{sec:implementation}

We have implemented \lang into a trans-compiler\footnote{Publicly available at \url{https://github.com/cmla-psu/shadowdp}.}  in Python. \lang currently
supports trans-compilation from annotated C code to target C code.
Its workflow includes two phases: \emph{transformation} and
\emph{verification}. The annotated source code will be checked and transformed
by \lang; the transformed code is further sent to a verifier.

\paragraph{Transformation} 
As explained in Section \ref{sec:typing}, \lang
tracks the typing environments in a flow-sensitive way, and
instruments corresponding statements when appropriate. Moreover, \lang adds an
assertion $\cod{assert}~(\vpriv{\priv} \le \priv)$ before the $\cod{return}$
command. This assertion specifies the final goal of proving differential
privacy.  The implementation follows the typing rules explained in Section~\ref{sec:typing}.

\paragraph{Verification}
The goal of verification is to prove the assertion $\cod{assert}~(\vpriv{\priv}
\le \priv)$ never fails for any possible inputs that satisfy the precondition
(i.e., the adjacency requirement). To demonstrate the usefulness of the transformed
programs, we use a model checker CPAChecker~\cite{beyer2011cpachecker}
v1.8. CPAChecker is capable of automatically verifying C program with a given
configuration. In our implementation, \emph{predicate analysis} is used. Also,
CPAChecker has multiple solver backends such as
MathSat~\cite{cimatti2013mathsat5}, Z3~\cite{z3} and
SMTInterpol~\cite{christ2012smtinterpol}. For the best
performance, we concurrently use different solvers and
return the results as soon as any one of them verifies the program.  

One limitation of CPAChecker and many other tools,
is the limited support for non-linear arithmetics. For programs with non-linear
arithmetics, we take two approaches. First, we verify the algorithm variants
where $\epsilon$ is fixed (the approach taken
in~\cite{Aws:synthesis}). In this case, all transformed code in our evaluation
is directly verified without any modification. Second, to verify the
correctness of algorithms with arbitrary $\epsilon$, we slightly rewrite the
non-linear part in a linear way or provide loop invariants (see Section
\ref{sec:difference_sparse_vector}). We report the results from both cases
whenever we encounter this issue.

\input{casestudies}

\subsection{Experiments}
\label{sec:evaluation}

\lang is evaluated on Report Noisy Max algorithm
(Figure~\ref{alg:transformed_noisymax}) along with all the algorithms discussed in Section~\ref{sec:casestudy}, as well as Partial Sum, Prefix Sum and Smart Sum
algorithms that are included in the \appendixref. For comparison, all the
algorithms verified in~\cite{Aws:synthesis} are included in the experiments (where Sparse Vector Technique is called Above Threshold in~\cite{Aws:synthesis}). One exception is \texttt{ExpMech} algorithm, since \lang currently lacks a sampling
command for \texttt{Exponential} noise. However, as shown in~\cite{lightdp}, it
should be fairly easy to add a noise distribution without affecting the
rest of a type system. 

Experiments are
performed on a Dual $\text{Intel}^{\text{\textregistered}}$
$\text{Xeon}^{\text{\textregistered}}$ E5-2620 v4@2.10GHz CPU machine with 64
GB memory. All algorithms are successfully checked and transformed by \lang and
verified by CPAChecker. For programs with non-linear arithmetics, we performed experiments on both solutions discussed in
Section~\ref{sec:difference_sparse_vector}. Transformation and verification
all finish within 3 seconds, as shown in Table \ref{tab:experiment_time},
indicating the simplicity of analyzing the transformed program, as well as the
practicality of verifying $\priv$-differentially private algorithms with \lang.

\subsection{Proof Automation}
\label{sec:automation}
\lang requires two kinds of annotations: (1) function specification and (2)
annotation for sampling commands. As most verification tools, (1) is
required since it specifies the property being verified. In all of our verified
examples, (2) is fairly simple and easy to write. To further improve the
usability of \lang, we discuss some heuristics to automatically generate the
annotations for sampling commands. 

Sampling commands requires two parts of annotation:
\begin{enumerate}
\item \textbf{Selectors}. The selector has two options: aligned ($\alignd$) or shadow ($\shadow$), with potential dependence. The heuristic is to enumerate branch conditions. For Report Noisy Max, there is only one branch condition $\Omega$, giving us four possibilities: $\alignd\ /\ \shadow\ /\  \Omega\mathbin{?}\alignd\mathbin{:}\shadow\ /\  \Omega\mathbin{?}\shadow\mathbin{:}\alignd$.

\item \textbf{Alignments for the sample}. 
It is often simple arithmetic on a small integer such as 0, 1, 2 or the exact difference of query answers and other program variables. For dependent types, we can also use the heuristic of using branch conditions. For Report Noisy Max, this will discover the correct alignment $\Omega\mathbin{?}2\mathbin{:}0$.
\end{enumerate}
This enables the discovery of all the correct annotations for the algorithm
studied in this paper. We leave a systematic study of proof automation as
future work.

%% file: casestudies.tex
\subsection{Case Studies}\label{sec:casestudy}

We investigate some interesting differentially private
algorithms that are formally verified by \lang. We only present the most
interesting programs in this section; the rest are provided in the \appendixref.

\subsubsection{Sparse Vector Technique}

Sparse Vector Technique~\cite{diffpbook} is a powerful mechanism which has been proven to satisfy $\priv$-differential privacy (its proof is notoriously tricky to write manually \cite{ninghuisparse}). In this section we show how \lang verifies this algorithm and later show how a novel variant is verified.

\begin{figure}[!ht]
\small
\setstretch{0.8}
\raggedright
\noindent\rule{\linewidth}{0.8pt}
\funsigfour{SVT}{$\priv$, $\cod{size}$, $\cod{T}$, $\cod{N}\annotation{:\tyreal_{\pair 0 0}}$; $\cod{q}\annotation{:\tylist~\tyreal_{\pair * *}}$}{($\cod{out}\annotation{:\tylist~\bool}$)}{$\alldiffer$}
\algrule
\begin{lstlisting}[frame=none,escapechar=@]
$\eta_1$ := $\lapm{(2/\priv)}\sampleannotation{\alignd}{1}$;
$\tT$ := T + $\eta_1$; count := 0; i := 0;
while (count < N $\land$ i < size)
  $\eta_2$ := $\lapm{(4N/\priv)}\sampleannotation{\alignd}{\Omega\mathbin{?}2:0}$;@\label{line:svt_distance}@
  if (q[i] + $\eta_2\geq\tT$) then
    out := true::out;@\label{line:svt_output}@
    count := count + 1;
  else
    out := false::out;
  i := i + 1;
\end{lstlisting}
The transformed program (slightly simplified for readability), where underlined commands are added by the type system:
\begin{lstlisting}[frame=none]
$\instrument{\vpriv{\epsilon}\text{ := 0;}}$
$\instrument{\havoc{\eta_1}\text{; }\vpriv{\epsilon}\text{ := }\vpriv{\epsilon}\text{ + }\priv/2\text{;}}$
$\tT$ := T + $\eta_1$; count:= 0; i := 0;
while (count < N $\land$ i < size)
  $\instrument{\assertcmd(\text{count < N} \land \text{i < size})\text{;}}$
  $\instrument{\havoc{\eta_2}\text{; }\vpriv{\priv}=\Omega\mathbin{?}(\vpriv{\priv} + 2 \times \priv / 4N):(\vpriv{\priv}\text{ + 0})\text{;}}$
  if (q[i] + $\eta_2\geq\tT$) then
    $\instrument{\assertcmd(\text{q[i]} + \first{\distance{\text{q}}}\text{[i]}+\eta_2+2\geq\tT +1)\text{;}}$
    out := true::out;
    count := count + 1;
  else
    $\instrument{\assertcmd(\lnot(\text{q[i]} + \first{\distance{\text{q}}}\text{[i]}+\eta_2 \ge \tT + 1 ))\text{;}}$
    out := false::out;
  i := i + 1;
\end{lstlisting}
\noindent\rule{\linewidth}{0.8pt}
\caption{Verifying Sparse Vector Technique with \lang (slightly simplified for readability). Annotations are in gray where $\Omega$ represents the branch condition.}
\label{alg:sparsevector}
\end{figure}

Figure~\ref{alg:sparsevector} shows the pseudo code of Sparse Vector Technique~\cite{diffpbook}. It examines the input queries and reports whether each query is above or below a threshold $T$. \begin{savenotes}
\begin{table*}[!ht]
\caption{Time spent on type checking and verification\label{tab:experiment_time}}
\vspace{-0.75em}
\small
\begin{center}
\begin{tabular}{c c c c c}
\Xhline{2\arrayrulewidth}
\textbf{Algorithm} & \textbf{Type Check (s)}   & \multicolumn{2}{c}{\textbf{Verification by \lang (s)}} & \textbf{Verification by~\cite{Aws:synthesis} (s)} \\
\hline
Report Noisy Max & 0.465 & \multicolumn{2}{c}{1.932} & 22\\
Sparse Vector Technique ($N=1$) & 0.398 & \multicolumn{2}{c}{1.856} & 27\\
\cline{3-4} & 
& \textbf{Rewrite} & \textbf{Fix $\epsilon$} &  \\
\cline{3-4}
Sparse Vector Technique & 0.399 & 2.629 & 1.679 & 580 \\
Numerical Sparse Vector Technique ($N=1$) & 0.418 & 1.783 & 1.788 & 4\\
Numerical Sparse Vector Technique & 0.421 & 2.584 & 1.662 & 5\\
Gap Sparse Vector Technique & 0.424 & 2.494 & 1.826 & N/A\\
Partial Sum & 0.445 & 1.922 & 1.897 & 14\\
Prefix Sum & 0.449 & 1.903 & 1.825 & 14\\
Smart Sum & 0.603 & 2.603 & 2.455 & 255\\
\Xhline{2\arrayrulewidth}
\end{tabular}
\end{center}
\end{table*}
\end{savenotes}To achieve differential privacy, it first adds Laplace noise to the threshold $T$, compares the noisy query answer $q[i]+\eta_2$ with the noisy threshold $\tT$, and returns the result ($\true$ or $\false$). The number of $\true$'s the algorithm can output is bounded by argument $N$. One key observation is that once the noise has been added to the threshold, outputting $\false$ pays no privacy cost~\cite{diffpbook}. As shown in Figure~\ref{alg:sparsevector}, programmers only have to provide two simple annotations: $\alignd\mathbin{,}1$ for $\eta_1$ and $\alignd\mathbin{,}\Omega\mathbin{?}2:0$ for $\eta_2$. Since the selectors in this example only select aligned version of variables, the shadow execution is optimized away (controlled by $pc$ in rule~\ruleref{T-If}). \lang successfully type checks and transforms this algorithm. However, due to a nonlinear loop invariant that CPAChecker fails to infer, it fails to verify the program. With the loop invariant provided manually, the verification succeeds, proving this algorithm satisfies $\priv$-differential privacy (we also verified a variant where $\epsilon$ is fixed to $N$ to remove the non-linearity).

\subsubsection{Gap Sparse Vector Technique}
\label{sec:difference_sparse_vector}
We now consider a novel variant of Sparse Vector Technique. In this variant, whenever $\mathtt{q[i]} + \eta_2 \geq \tT$, it outputs the value of the gap $\mathtt{q[i]} + \eta_2-\tT$ (how much larger the noisy answer is compared to the noisy threshold). Note that the noisy query value $\mathtt{q[i]} + \eta_2$ is reused for both this check and the output (whereas other proposals either (1) draw fresh noise and result in a larger $\epsilon$ \cite{diffpbook}, or (2) re-use the noise but do not satisfy differential privacy, as noted in \cite{ninghuisparse}). For noisy query values below the noisy threshold, it only outputs $\false$. We call this algorithm GapSparseVector. More specifically, Line~\ref{line:svt_output} in Figure~\ref{alg:sparsevector} is changed from \texttt{out := true::out;} to the following: \texttt{out := (q[i] + $\eta_2$ - $\tT$)::out;}. To the best of our knowledge, the correctness of this variant has not been noticed before.
This variant can be easily verified with little changes to the original annotation. One observation is that, to align the \texttt{out} variable, the gap appended to the list must have $0$ aligned distance. Thus we change the distance of $\eta_2$ from $\Omega\mathbin{?}2\mathbin{:}0$ to $\Omega\mathbin{?}(1-\first{\distance{\cod{q}}}\cod{[i]}):0$, the other part of the annotation remains the same.

\lang successfully type checks and transforms the program. Due to the non-linear arithmetics issue, we
rewrite the assignment command $\vpriv{\priv}\text{  := } \vpriv{\priv} + (1-\first{\distance{\cod{q}}}\cod{[i]})
\times \priv / 4N\texttt{;}$ to
$\assertcmd(|1-\first{\distance{\cod{q}}}\cod{[i]}| \le
2)\text{;}~\vpriv{\priv}\text{ := }\vpriv{\priv} + 2\times\priv / 4N\texttt{;}$ and provide nonlinear loop invariants; then it is verified  (we also verified a variant where $\epsilon$ is fixed to 1).

%% file: relatedwork.tex
\paragraph{Randomness alignment based proofs} 
The most related work is LightDP~\cite{lightdp}. \lang is inspired by LightDP
in a few aspects, but also with three significant differences. First, \lang
supports shadow execution, a key enabling technique for the verification of
Report Noisy Max based on standard program semantics. Second, while LightDP has
a flow-insensitive type system, \lang is equipped with a flow-sensitive one.
The benefit is that the resulting type system is both more expressive and more
usable, since only sampling command need annotations. Third, \lang
allows extra permissiveness of allowing two related executions to take
different branches, which is also crucial in verifying 
Report Noisy Max.  In fact, \lang is strictly more expressive than LightDP: LightDP is a restricted form of \lang where the shadow execution is never used (i.e., when the selector always picks the aligned execution).

\paragraph{Coupling based proofs} 
The state-of-the-art verifier based on approximate
coupling~\cite{Aws:synthesis} is also able to verify the algorithms we have
discussed in this paper. Notably, it is able to automatically verify proofs for
algorithms including Report-Noisy-Max and Sparse Vector. However, verifying the
transformed program by \lang is significantly easier than verifying the
first-order Horn clauses and probabilistic constraints generated by their tool.
In fact, \lang verifies all algorithms within 3 seconds while the coupling
verifier takes 255 seconds in verifying Smart Sum and 580 seconds in verifying
Sparse Vector (excluding proof synthesis time). Also, instead of building the
system on \emph{customized} relational logics to verify differential
privacy~\cite{Barthe12,EasyCrypt,BartheICALP2013,Barthe16,BartheCCS16}, \lang
bases itself on \emph{standard} program logics, which makes the transformed
program re-usable by other program analyses. 

\paragraph{Other language-based proofs}
Recent work such as Personalized Differential Privacy
(PDP)~\cite{EbadiPOPL2015} allows each individual to set its own different
privacy level and PDP will satisfy difference privacy regarding the level she
sets. PINQ~\cite{pinq} tracks privacy consumption dynamically on databases and
terminate when the privacy budget is exhausted. However, along with other work
such as computing bisimulations families for probabilistic
automata~\cite{Tschantz11,Xu2014}, they fail to provide a tight bound on the
privacy cost of sophisticated algorithms.
\vspace{-1em}

%% file: conclusions.tex
In this paper we presented \lang, a new language for the verification of differential privacy algorithms. \lang uses shadow execution to generate more flexible randomness alignments that allows it to verify more algorithms, such as Report Noisy Max, than previous work based on randomness alignments. We also used it to verify a novel variant of Sparse Vector that reports the gap between noisy above-threshold queries and the noisy threshold. 

Although \lang only involves minimum annotations, one future work is to fully
automate the verification using \lang, as sketched in
Section~\ref{sec:automation}. Another natural next step is to extend \lang to
support more noise distributions, enabling it to verify more algorithms such as
\texttt{ExpMech} which uses \texttt{Exponential} noise. Furthermore, we plan 
to investigate other applications of the transformed program. For instance,
applying symbolic executors and bug finding tools on the transformed program to
construct counterexamples when the original program is buggy.

%% file: appendix.tex
\clearpage
\appendix
\nobalance
\addcontentsline{toc}{section}{Appendices}
\section*{Appendix}

\section{\lang Semantics}

\input{semantics}

Given a distribution $\mu \in \dist(A)$, its support is defined as
$\support(\mu)\defn \{a\mid \mu(a)>0\}$.  We use $\dgdist_a$ to represent the
degenerate distribution $\mu$ that $\mu(a)=1$ and $\mu(a')=0$ if $a'\not=a$.
Moreover, we define monadic functions $\unitop$ and $\bindop$ functions to
formalize the semantics for commands:
\begin{align*}
\unitop &: A \rightarrow \dist(A) \defn \lambda a.~\dgdist_a \\ 
\bindop &: \dist(A)\rightarrow(A \rightarrow \dist(B))\rightarrow\dist(B)  \\
        &\defn \lambda \mu.~\lambda f.~(\lambda b.~\sum_{a\in A} (f~a~b)\times \mu(a))
\end{align*}
That is, $\unitop$ takes an element in $A$ and returns the Dirac distribution
where all mass is assigned to $a$; $\bindop$ takes $\mu$, a distribution on
$A$, and $f$, a mapping from $A$ to distributions on $B$ (e.g., a conditional
distribution of $B$ given $A$), and returns the corresponding marginal
distribution on $B$. This monadic view avoids cluttered definitions and proofs
when probabilistic programs are involved. 

Figure~\ref{fig:semantics} provides the semantics of commands.

\section{Constructing Shadow Execution}
\input{shadowrules}

Figure~\ref{fig:shadowrules_expression} shows the rule for generating shadow execution expressions as well as aligned execution expressions. Figure~\ref{fig:shadowrules} shows the rules for generating shadow execution
commands. The shadow execution essentially replaces each variables $x$ with
their correspondence (i.e., $x+\Gamma(x)$) in $c$, as standard in
self-composition construction~\cite{barthe2004, terauchi2005}. Compared with
standard self-composition, the differences are:
\begin{enumerate}[leftmargin=5mm]
\item $\shadowexec{c,\Gamma}$ is not applicable to sampling commands, since if
the original execution takes a sample while the shadow execution does not, we
are unable to align the sample variable due to different probabilities.  

\item For convenience, we use $x + \second{\nexpr}$ where $\Gamma  \proves  x :
\pair {\first\nexpr} {\second\nexpr}$ whenever the shadow value of $x$ is used;
correspondingly, we update $\second{\distance{x}}$ to $v-x$ instead of updating
the shadow value of $x$ to some value $v$.
\end{enumerate}

\section{Extra Case Studies}
In this section, we study extra differentially private algorithms to show the
power of \lang. As Section~\ref{sec:casestudy}, the shadow execution part is
optimized away when the selectors never use the shadow variables.

\subsection{Numerical Sparse Vector Technique}
Numerical Sparse Vector Technique~\cite{diffpbook} is an interesting variant of
Sparse Vector Technique which outputs numerical query answers when the query
answer is large. To achieve differential privacy, like Sparse Vector Technique,
it adds noise to the threshold $T$ and each query answer $\mathtt{q[i]}$; it
then tests if the noisy query answer is above the noisy threshold or not. The
difference is that Numerical Sparse Vector Technique draws a fresh noise
$\eta_3$ when the noisy query answer is above the noisy threshold, and then
releases $\mathtt{q[i]} + \eta_3$ instead of simply releasing $\true$. The
pseudo code for this algorithm is shown in Figure~\ref{alg:numsparsevector}.

In this algorithm, \lang needs an extra annotation for the new sampling command
of $\eta_3$. We use the same approach in Gap Sparse Vector
Technique for this new sampling command. Recall the observation that we want the
final output variable \texttt{out} to have distance $\pair 0 -$, which implies
that the numerical query $\mathtt{q[i]} + \eta_3$ should have distance $\pair 0
-$. We can deduce that $\eta_3$ must have distance
$-\first{\distance{\cod{q}}}\cod{[i]}$ inside the branch, thus we write
``$\alignd, -\first{\distance{\cod{q}}}\cod{[i]}$'' for $\eta_3$. The rest of
the annotations remain the same as standard Sparse Vector Technique. 

\begin{figure}[ht]
\small
\setstretch{0.8}
\raggedright
\noindent\rule{\linewidth}{0.8pt}
\funsigfour{NumSVT}{$\priv$, $\cod{size}$, $\cod{T}$, $\cod{N}$, \annotation{:\tyreal_{\pair 0 0}}; $\cod{q}\annotation{:\tylist~\tyreal_{\pair * *}}$}{($\cod{out}\annotation{:\tylist~\tyreal_{\pair 0 -}}$)}{$\alldiffer$}
\algrule
\begin{lstlisting}[frame=none,escapechar=@]
$\eta_1$ := $\lapm{(3/\priv)}\sampleannotation{\alignd}{1}$;
$\tT$ := $T + \eta_1$; count := 0; i := 0;
while (count < N $\land$ i < size)
  $\eta_2$ := $\lapm{(6N/\priv)}\sampleannotation{\alignd}{\Omega\mathbin{?}2\mathbin{:}0}$;
  if (q[i] + $\eta_2\geq\tT$) then
    $\eta_3$ := $\lapm{(3N/\priv)}\sampleannotation{\alignd}{-\first{\distance{\text{q}}}\text{[i]}}$;@\label{line:numsvt_eta_3}@
    out := (q[i] + $\eta_3$)::out;
    count := count + 1;
  else
    out := 0::out;
  i := i + 1;
\end{lstlisting}
The transformed program (slightly simplified for readability), where underlined commands are added by the type system:
\begin{lstlisting}[frame=none]
$\instrument{\vpriv{\epsilon}\text{ := 0;}}$
$\instrument{\havoc{\eta_1}\text{;}~\vpriv{\epsilon} := \vpriv{\epsilon} + \priv/3\text{;}}$
$\tT := T + \eta_1\text{;}$
count := 0; i := 0;
while (count < N $\land$ i < size)
  $\instrument{\assertcmd(\text{count < N}~\land~\text{i < size})\text{;}}$
  $\instrument{\havoc{\eta_2}\text{;}~\vpriv{\priv}=\Omega\mathbin{?}(\vpriv{\priv} + 2 \times \priv / 6N):(\vpriv{\priv}\text{ + 0})\text{;}}$
  if ($\text{q[i]}+\eta_2\geq \tT$) then
    $\instrument{\assertcmd(\text{q[i] + } \first{\distance{\text{q}}}\text{[i]}+\eta_2+2\geq\tT +1)\text{;}}$
    $\instrument{\havoc{\eta_3}\text{;}~\vpriv{\priv}=\vpriv{\priv}+|-\first{\distance{\text{q}}}\text{[i]}|\times \priv/3N\text{;}}$
    out := (q[i] + $\eta_3$)::out;
    count := count + 1;
  else
    $\instrument{\assertcmd(\lnot(\text{q[i]} + \first{\distance{\text{q}}}\text{[i]}+\eta_2 \ge \tT +1 ))\text{;}}$
    out := 0::out;
  i := i + 1;
\end{lstlisting}
\noindent\rule{\linewidth}{0.8pt}
\caption{Verifying Numerical Sparse Vector Technique with \lang. Annotations are in gray where $\Omega$ represents the branch condition.}
\label{alg:numsparsevector}
\end{figure}

For the non-linear issue of the verifier, we rewrite the privacy cost assignment from 
$$\vpriv{\priv} = \vpriv{\priv} + |-\first{\distance{\cod{q}}}\cod{[i]}| \times \priv / 3;$$
$$\text{to}\quad\assertcmd(|-\first{\distance{\cod{q}}}\cod{[i]}| \le 1)\text{;}\ \vpriv{\priv} = \vpriv{\priv} + \priv / 3\text{;}$$

With manual loop invariants provided, CPAChecker successfully verified the rewritten program.

\subsection{Partial Sum}

We now study an $\priv$-differentially private algorithm PartialSum
(Figure~\ref{alg:partialsum}) which simply sums over a list of queries. To
achieve differential privacy, it adds noise using Laplace mechanism to the
final sum and output the noisy sum. One difference from the examples in
Section~\ref{sec:casestudy} is the adjacency assumption: at most one query
answer may differ by 1 as specified in the precondition.

In this example, since the noise is added to the final \texttt{sum}, it makes no difference if we choose the aligned version or shadow version of normal variables (they are both identical to the original execution). To decide the distance of the random variable, we want the final output to have distance $\pair 0 -$, it is easy to deduce that the distance of $\eta$ should be $-\first{\distance{\cod{sum}}}$. Adding the annotation $\alignd,\ -\first{\distance{\cod{sum}}}$ to Line~\ref{line:partialsum_eta} in Figure~\ref{alg:partialsum}, \lang successfully transforms and verifies the program. Note that since the cost update command $\vpriv{\priv} := \vpriv{\priv} + |\distance{\cod{sum}}| \times \priv\text{;}$ contains non-linear arithmetic, we carefully rewrite this command to $
\cod{assert}~(|\distance{\cod{sum}}| \le 1);\ \vpriv{\priv} := \vpriv{\priv} + \priv;$. \lang is able to type check and verify this algorithm within seconds.

\begin{figure}[ht]
\raggedright
\setstretch{0.85}
\small
\noindent\rule{\linewidth}{2\arrayrulewidth}
\funsigfour{PartialSum}{$\priv,\cod{size}\annotation{:\tyreal_{\pair 0 0}}; \cod{q}\annotation{:{\tylist~\tyreal_{\pair * *}}} $}{($\cod{out}\annotation{:\tyreal_{\pair 0 -}}$)}{$\onediffer$}
\algrule
\begin{lstlisting}[frame=none,escapechar=@]
sum := 0; i := 0;
while (i < size)
  sum := sum + q[i];
  i := i + 1;
$\eta$ = $\lapm{(1/\priv)}\sampleannotation{\alignd}{-\first{\distance{\text{sum}}}}$;@\label{line:partialsum_eta}@
out := sum + $\eta$;
\end{lstlisting}
The transformed program, where underlined commands are added by the type system:
\begin{lstlisting}[frame=none]
sum := 0; i := 0;
$\instrument{\first{\distance{\text{sum}}}\text{ := 0;}}$
while (i < size)
  $\instrument{\assertcmd(\text{i} < \text{size})\text{;}}$
  sum := sum + q[i];
  $\instrument{\first{\distance{\text{sum}}}\text{ := }\first{\distance{\text{sum}}}\text{ + } \first{\distance{\text{q}}}\text{[i]}\text{;}}$
  i := i + 1;
$\instrument{\havoc{\eta}\text{;}~\vpriv{\priv}\text{ := }\vpriv{\priv} + |\first{\distance{\text{sum}}}|\times\priv\text{;}}$
out := sum + $\eta$;
\end{lstlisting}
\noindent\rule{\linewidth}{2\arrayrulewidth}
\caption{Verifying Partial Sum using \lang. Annotations are shown in gray.}
\label{alg:partialsum}
\end{figure}

\subsection{Smart Sum and Prefix Sum}
Another interesting algorithm SmartSum~\cite{chan10continual} has been previously verified~\cite{Barthe12,Barthe14} with heavy annotations. We illustrate the power of our type system by showing that this algorithm can be verified with very little annotation burden for the programmers. 

\begin{figure}[ht]
\raggedright
\setstretch{0.85}
\small
\noindent\rule{\linewidth}{2\arrayrulewidth}
\funsigfour{SmartSum}{$\priv$, $\cod{M}$, $\cod{T}\annotation{:\tyreal_{\pair 0 0}}$ $\cod{q}\annotation{:\tylist~\tyreal_{\pair * *}}$}{($\cod{out}\annotation{:\tylist~\tyreal_{\pair 0 -}}$)}{$\onediffer$}
\algrule
\begin{lstlisting}[frame=none, escapechar=@]
next:= 0; i:= 0; sum:= 0;
while i $\leq$ T
  if (i + 1) mod M = 0 then @\label{line:smartsum_if}@
    $\eta_1$ := $\lapm{(1/\priv)}\sampleannotation{\alignd}{-\first{\distance{\text{sum}}}-\first{\distance{\text{q}}}\text{[i]}}$;@\label{line:smartsum_eta1}@
    next:= sum + q[i] + $\eta_1$;
    sum := 0;
    out := next::out; 
  else @\label{line:smartsum_else}@
    $\eta_2$ := $\lapm{(1/\priv)}\sampleannotation{\alignd}{-\first{\distance{\text{q}}}\text{[i]}}$;@\label{line:smartsum_eta2}@
    next:= next + q[i] + $\eta_2$;
    sum := sum + q[i];
    out := next::out;
  i := i + 1;
\end{lstlisting}
The transformed program:
\begin{lstlisting}[frame=none, escapechar=@]
next:=0; n:=0; i:=0; sum := 0;
$\instrument{\first{\distance{\text{sum}}}\text{ := 0;}}$
while (i < size $\land$ i $\leq$ T )
  $\instrument{\assertcmd(\text{i < size}~\land~ \text{i} \leq \text{T})\text{;}}$
  if (i + 1) mod M = 0 then
    $\instrument{\havoc{\eta_1}\text{;}~\vpriv{\priv}=\vpriv{\priv}\ \text{+}\ |-\first{\distance{\text{sum}}}-\first{\distance{\text{q}}}\text{[i]}|\times\priv\text{;}}$@\label{line:smartsum_transformed_eta1}@
    next:= sum + q[i] + $\eta_1$;
    sum := 0;
    out := next::out; 
    $\instrument{\first{\distance{\text{sum}}}\text{ := 0;}}$
  else
    $\instrument{\havoc{\eta_2}\text{;}~\vpriv{\priv}=\vpriv{\priv}\ \text{+}\ |-\first{\distance{\text{q}}}\text{[i]}|\times\priv\text{;}}$@\label{line:smartsum_transformed_eta2}@
    next:= next + q[i] + $\eta_2$;
    sum := sum + q[i];
    out := next::out;
    $\instrument{\first{\distance{\text{sum}}}\text{ := }\first{\distance{\text{sum}}} +\first{\distance{\text{q}}}\text{[i];}}$
  i := i + 1;
\end{lstlisting}
\noindent\rule{\linewidth}{2\arrayrulewidth}
\caption{Verifying SmartSum algorithm with \lang. Annotations are shown in gray.}
\label{alg:smartsum}
\end{figure}

This algorithm (Figure~\ref{alg:smartsum}) is designed to continually release aggregate statistics in a privacy-preserving manner. The annotations are only needed on Line~\ref{line:smartsum_eta1} and \ref{line:smartsum_eta2} in Figure~\ref{alg:smartsum}. We use the same observation as stated in Section~\ref{sec:difference_sparse_vector}, in order to make the aligned distance of the output variable \texttt{out} 0. To do that, we assign distance $-\first{\distance{\cod{sum}}}-\first{\distance{\cod{q}}}\cod{[i]}$ to $\eta_1$ and $-\first{\distance{\cod{q}}}\cod{[i]}$ to $\eta_2$, and use $\alignd$ for both random variables. This algorithm is successfully transformed by \lang. However, due to the non-linear issue described in Section~\ref{sec:implementation}, we change the commands of Line~\ref{line:smartsum_transformed_eta1} and Line~\ref{line:smartsum_transformed_eta2} from 
\begin{align}
\vpriv{\priv} &:= \vpriv{\priv} + |-\first{\distance{\cod{sum}}}-\first{\distance{\cod{q}}}\cod{[i]}|\times \priv\texttt{;}\nonumber \\
\vpriv{\priv} &:= \vpriv{\priv} + |-\first{\distance{\cod{q}}}\cod{[i]}|\times \priv\texttt{;} \nonumber 
\end{align}
to

\begin{lstlisting}[frame=none, numbers=none]
if ($|-\first{\distance{\cod{sum}}} - \first{\distance{\cod{q}}}\cod{[i]}| > 0$) 
   $\assertcmd(|-\first{\distance{\cod{sum}}} - \first{\distance{\cod{q}}}\cod{[i]}|\le 1)\texttt{;}~\vpriv{\priv} := \vpriv{\priv} + \priv$;
if ($|-\first{\distance{\cod{q}}}\cod{[i]}| > 0$)
   $\assertcmd(|-\first{\distance{\cod{q}}}\cod{[i]}| \le 1)\texttt{;}~\vpriv{\priv} := \vpriv{\priv} + \priv$;
\end{lstlisting}

Moreover, one difference of this algorithm is that it satisfies $2\priv$-differential privacy~\cite{chan10continual} instead of $\priv$-differential privacy, thus the last assertion added to the program is changed to $\cod{assert}~(\vpriv{\priv} \le 2\times\priv)\cod{;}$. With this CPAChecker is able to verify this algorithm.

We also verified a variant of Smart Sum, called Prefix Sum
algorithm~\cite{Aws:synthesis}. This algorithm is a simplified and less precise
version of Smart Sum, where the else branch is always taken. More specifically,
we can get Prefix Sum by removing Lines~\ref{line:smartsum_if} -
\ref{line:smartsum_else} from Figure~\ref{alg:smartsum}. The annotation remains
the same for $\eta_2$, type checking and transformation follows Smart Sum. Note
that Prefix Sum satisfies $\epsilon$-differential privacy, so the last
assertion remains unchanged.  CPAChecker then verifies the transformed Prefix
Sum within 2 seconds.

\section{Soundness Proof}

We first prove a couple of useful lemmas. Following the same notations as in
Section \ref{sec:soundness}, we use $m$ for original memory and $m'$ for
extended memory with aligned and shadow variables but not the distinguished
privacy tracking variable $\vpriv{\priv}$. Moreover, we assume that memory
tracks the entire list of sampled random values at each point.

\begin{lemma}[Numerical Expression]
\label{lem:aexpr}
$\forall e,m', \Gamma$, such that $\Gamma \proves e : \tyreal_{\pair{\first{\nexpr}}{\second{\nexpr}}}$,
we have
\[
\evalexpr{e}{\store'} + \evalexpr{\first{\nexpr}}{\store'} = \evalexpr{e}{\amem{\Gamma}{m'}},\]
\[
\evalexpr{e}{\store'} + \evalexpr{\second{\nexpr}}{\store'} = \evalexpr{e}{\smem{\Gamma}{m'}}.
\]
\end{lemma}
\begin{proof}
Induction on the inference rules.
\begin{itemize}
\item \ruleref{T-Num}: trivial.

\item \ruleref{T-Var}: If $\first{\Gamma}(x)$ is $*$, $\first{\nexpr}$ is
$\first{\distance{x}}$. By Definition~\ref{def:relation},
${\amem{\Gamma}{m'}}(x)=m'(x)+m'(\first{\distance{x}})$. If
$\first{\Gamma}(x)=\first{\dexpr}$, $\first{\nexpr}$ is $\first{\dexpr}$.
Again, we have ${\amem{\Gamma}{m'}}(x)=m'(x)+\evalexpr{\first{\dexpr}}{m'}$ by
Definition~\ref{def:relation}. The case for $\second{\Gamma}$ is similar.

\item \ruleref{T-OPlus, T-OTimes, T-Ternary, T-Index}: by the induction
hypothesis (list is treated as a collection of variables of the same type).
\end{itemize}
\end{proof}

\begin{lemma}[Boolean Expression]
\label{lem:bexpr}
$\forall e,m', \Gamma$, such that $\Gamma \proves e : \bool$,
we have
\[
\evalexpr{e}{\store'}  = \evalexpr{e}{\first{\Gamma}\store'}  = \evalexpr{e}{\second{\Gamma}\store'}.
\]
\end{lemma}
\begin{proof}
Induction on the inference rules.

\begin{itemize}
\item \ruleref{T-Boolean}: trivial.

\item \ruleref{T-Var}: the special case where
$\first{\dexpr}=\second{\dexpr}=0$. Result is true by
Definition~\ref{def:relation}.

\item \ruleref{T-ODot}: by Lemma~\ref{lem:aexpr} and induction hypothesis, we
have $\evalexpr{e_i}{\store'} + \evalexpr{\first{\nexpr_i}}{\store'} =
\evalexpr{e_i}{\amem{\Gamma}{m'}}$ for $i\in \{1,2\}$. By the assumption of
\ruleref{T-ODot}, we have \[\evalexpr{e_1\odot e_2}{m'}\Leftrightarrow
\evalexpr{(e_1+\first{\nexpr_1}) \odot (e_2+\first{\nexpr_2})}{m'}=
\evalexpr{e_1\odot e_2}{\amem{\Gamma}{m'}}\] The case for $\second{\Gamma}$ is
similar.

\item \ruleref{T-Ternary, T-Index}: by the induction hypothesis.
\end{itemize}
\end{proof}

\subsection{Injectivity}

We first prove that the type system maintains injectivity.

\begin{lemma}
\label{lem:inj}
$\forall c,c',\pc, m',m_1', m_2', \Gamma_1, \Gamma_2.~\pc \proves \cmdfull{\Gamma_1}{\Gamma_2}{c\transform c'} \land \evalexpr{c'}{m'}{m_1'}\not=0 \land \evalexpr{c'}{m'}{m_2'}\not=0$, then we have
\begin{align*}
(m_1' = m_2') \lor (\exists \eta.~\Gamma_2^\star m_1'(\eta)  \not=
\Gamma_2^\star m_2'(\eta)), \star\in \{\alignd, \shadow\}
\end{align*}
\end{lemma}

\begin{proof}
By structural induction on $c$. 

\begin{itemize}[leftmargin=5mm]
\item Case $\skipcmd$: trivial since it is non-probabilistic.

\item Case $x:=e$: trivial since it is non-probabilistic.  

\item Case $c_1;c_2$: Let $\pc \proves \Gamma_1 \{c_1 \transform c_1' \} \Gamma$,
$\pc \proves \Gamma \{c_2 \transform c_2'\} \Gamma_2$. There exists some $m_3'$
and $m_4'$ such that

\begin{align*}
\evalexpr{c_1'}{m'}(m_3')\not=0 &\land \evalexpr{c_2'}{m_3'}(m_1')\not=0\\
\evalexpr{c_1'}{m'}(m_4')\not=0 &\land \evalexpr{c_2'}{m_4'}(m_2')\not=0
\end{align*}

If $m_3'=m_4'$, results is true by the induction hypothesis on $c_2'$.
Otherwise, we have $\exists \eta.~\Gamma_2^\star m_3'(\eta) \not=\Gamma_2^\star
m_4'(\eta)$. Hence, $\Gamma_2^\star m_1'(\eta) \not=\Gamma_2^\star m_2'(\eta)$;
contradiction.

\item Case $\ifcmd{e}{c_1}{c_2}$: let $\pc'\proves \Gamma_1 \{c_i \transform c_i'\}
\Gamma_i'$.  By typing rule, we have $\Gamma_1'\join \Gamma_2'=\Gamma_2$. Given
initial memory $m'$, both executions take the same branch. Without
loss of generality assume $c_1$ is executed. Let $\evalexpr{c_1'}{m'}{m_3'}\not=0$,
$\evalexpr{c_1''}{m_3'}{m_1'}\not=0$, $\evalexpr{c_1'}{m'}{m_4'}\not=0$ and
$\evalexpr{c_1''}{m_4'}{m_1'}\not=0$. There are two cases:

\begin{enumerate}
\item $m_3'=m_4'$: in this case, $c_1''$ only changes the value of $\third{x}$ to
$\nexpr$ when $\Gamma_2(x)=*$ and $\Gamma_1'(x)=\nexpr$. Moreover, we have
$m_1'(\third{x}) =\evalexpr{\nexpr}{m_3'}=\evalexpr{\nexpr}{m_4'}=
m_2'(\third{x})$ for those variables. Hence, $m_1'=m_2'$.

\item $m_3'\not=m_4'$: by induction hypothesis, $\exists \eta.~m_3'(\eta)\not=
m_4'(\eta)$. Hence, $m_1' \not= m_2'$, and $m_1'(\eta)\not=
m_2'(\eta)$ for the same $\eta$.
\end{enumerate}

When $\pc=\top$, $c'$ also includes $\second{c}$. In this case, result still
holds after $\second{c}$ since $\second{c}$ is deterministic by construction.

\item Case $\whilecmd{e}{c}$: let $\pc'\proves \Gamma\join \Gamma_f \{c
\transform c'\} \Gamma_f$. We proceed by induction on the number of iterations:
\begin{enumerate}
\item $c$ is not executed: $m_1'=m'=m_2'$.
\item $c$ is executed $n+1$ times: similar to the $c_1;c_2$ case. 
\end{enumerate}

When $\pc=\top$, $c'$ also includes $\second{c}$. In this case, result still
holds after $\second{c}$ since $\second{c}$ is deterministic by construction.

\item Case $(\eta := \lapm~\real;\select,\nexpr_\eta )$: if $m_1'\not=m_2'$, we
know that $m_1'(\eta) \not= m_2'(\eta)$ since $c'$ only modifies the value of
$\eta$.  In this case, it must be true that $\Gamma_2^\star m_1'(\eta)\not=
\Gamma_2^\star m_2'(\eta)$: otherwise, $m_1'=m_2'$ by the injectivity check in
rule~\ruleref{T-Laplace}.

\end{itemize}
\end{proof}

\paragraph{Proof of Lemma~\ref{lem:injective}}

\begin{proof}
Direct implication of Lemma~\ref{lem:inj}.
\end{proof}

\subsection{Instrumentation}
Next, we show a property offered by the auxiliary function  $\Gamma_1,
\Gamma_2, \pc \Rrightarrow c'$ used in the typing rules. Intuitively, they
allow us to freely switch from one typing environment to another by promoting
some variables to star type.

\begin{lemma}[Instrumentation]
\label{lem:triplearrow}
$\forall \pc, \Gamma_1, \Gamma_2, c'$,  if $(\Gamma_1, \Gamma_2, \pc) \Rrightarrow c'$, then
for any memory $m_1'$, there is a unique $m_2'$ such that $\evalexpr{c'}{m_1'}(m_2')\not=0$ and
\[
\begin{cases}
\first{\Gamma_1} m_1' = \first{\Gamma_2} m_2' \land  \second{\Gamma_1} m_1' = \second{\Gamma_2} m_2' &\text{ if } \pc=\bot \\
\first{\Gamma_1} m_1' = \first{\Gamma_2} m_2' &\text{ if } \pc=\top
\end{cases}
\]
\end{lemma}
\begin{proof}
By construction, $c'$ is deterministic. Hence, there is a unique $m_2'$ such that
$\evalexpr{c}{m_1'}(m_2')\not=0$. 

Consider any variable $x\in \Vars$.  By
the construction of $c'$, we note that $m_1'(x)=m_2'(x)$ and $\first{\distance{x}}$
differs in $m_1'$ and $m_2'$ only if $\first{\Gamma_1}(x)=\nexpr$ and
$\first{\Gamma_2}(x)=*$. In this case, $\first{\Gamma_1}
m_1'(x)=m_1'(x)+\evalexpr{\nexpr}{m_1'}=m_2'(x)+m_2'(\first{\distance{x}})=\first{\Gamma_2}
m_2' (x)$.   Otherwise, $\first{\Gamma_1}(x)=\first{\Gamma_2}(x)$ (since
$\Gamma_1\sqsubseteq \Gamma_2$). When
$\first{\Gamma_1}(x)=\first{\Gamma_2}(x)=*$, $\first{\Gamma_1}
m_1'(x)=m_1'(x)+m_1'(\first{\distance{x}})=m_2'(x)+m_2'(\first{\distance{x}})=\first{\Gamma_2} m_2'(x)$.  When
$\first{\Gamma_1}(x)=\first{\Gamma_2}(x)=\nexpr$ for some $\nexpr$,
$\first{\Gamma_1} m_1'(x) = m_1'(x)+\evalexpr{\nexpr}{m_1'} =
m_2'(x)+\evalexpr{\nexpr}{m_2'} = \first{\Gamma_2} m_2'(x)$. 

When $\pc=\bot$, the same argument applies to the case of
$\second{\Gamma_1}m_1'=\second{\Gamma_2}m_2'$.  
\end{proof}

\subsection{Shadow Execution}
Next, we show the main properties related to shadow execution. 

\begin{lemma}[]
\label{lem:shadowexpr}
$\forall e, \Gamma, m'$, if $e$ is well-typed under $\Gamma$, we have
\[\evalexpr{\shadowexec{e,\Gamma}}{m'} =
\evalexpr{e}{\second{\Gamma}m'}\]  
\end{lemma}
\begin{proof}
By structural induction on the $e$. 

\begin{itemize}[leftmargin=5mm]

\item $e$ is $r$, $\true$ or $\false$: trivial.

\item $e$ is $x$: if $\second{\Gamma}(x)=*$ (with base type of $\tyreal$),
$\evalexpr{\shadowexec{x,\Gamma}}{m'}=m'(x)+m'(\second{\distance{x}})=\evalexpr{x}{\second{\Gamma}m'}$.
When $\Gamma(x)=\tyreal_{\pair \dexpr {\second{\nexpr}}}$, we have
$\evalexpr{\shadowexec{x,\Gamma}}{m'} = \evalexpr{x+\second{\nexpr}}{m'} =
m'(x)+\evalexpr{\second{\nexpr}}{m'} = \evalexpr{x}{\second{\Gamma}m'}$ (by
Lemma~\ref{lem:aexpr}). Otherwise,
$\evalexpr{\shadowexec{x,\Gamma}}{m'}=m'(x)=\second{\Gamma} m'(x)$.

\item $e$ is $e_1~\op~e_2$: by induction hypothesis.

\item $e$ is $e_1[e_2]$: By the typing rule~\ruleref{T-Index}, $\Gamma \proves
e_2:\tyreal_{\pair 0 0}$. Hence, $\evalexpr{e_2}{\second{\Gamma}
m'}=\evalexpr{e_2}{ m'}$ by Lemma~\ref{lem:aexpr}.

Let $\Gamma \proves e_1: \tylist~\tyreal_{\pair {\dexpr} {\second{\dexpr}}}$.
When $e_1=q$ and $\second{\dexpr}=*$,
$\evalexpr{\shadowexec{q[e_2],\Gamma}}{m'}
=\evalexpr{q[e_2]+\second{\distance{q}}[e_2]}{m'}
=\evalexpr{q[e_2]}{\second\Gamma m'}$.
When $\dexpr=*$ and $e_1$ is not a variable, $\evalexpr{\shadowexec{e_1[e_2],\Gamma}}{m'}$ is
defined as $\evalexpr{\shadowexec{e_1,\Gamma}[e_2]}{m'}$. By induction
hypothesis, this is the same as $\evalexpr{e_1}{\second{\Gamma}
m'}\evalexpr{[e_2]}{m'}=\evalexpr{e_1[e_2]}{\second{\Gamma} m'}$. 

When $\dexpr=\second{\nexpr}$, $\evalexpr{\shadowexec{e_1[e_2],\Gamma}}{m'}$ is
defined as $\evalexpr{e_1[e_2]+\second{\nexpr}}{m'}$, which is the same as
$\evalexpr{e_1[e_2]}{\second{\Gamma} m'}$ (by Lemma~\ref{lem:aexpr}).

Otherwise, result is true by Lemma~\ref{lem:bexpr}.

\item $e$ is $e_1::e_2$ and $\neg e$: by induction hypothesis.

\item $e$ is $e_1?e_2\!:\!e_3$: by induction hypothesis, we have
$\evalexpr{\shadowexec{e_1,\Gamma}}{m'}=\evalexpr{e_1}{\second{\Gamma}m'}$.
Hence, the same element is selected on both ends. Result is true by
induction hypothesis.
\end{itemize}
\end{proof}

Next, we show that shadow execution simulates two executions on
$\second{\Gamma}$-related memories via two program executions.

\begin{lemma}[Shadow Execution]
\label{lem:shadowexec}
$\forall c, \second{c},
\Gamma.~(\forall x\in
\asgnvars(c).~\second{\Gamma}(x)=*)~\land~\shadowexec{c,\Gamma}= \second{c}$, we have
\[\forall m_1',m_2'.~\evalexpr{\second{c}}{m_1'}(m_2')  =
\evalexpr{c}{\second{\Gamma}m_1'}(\second{\Gamma} m_2').\]
\end{lemma}

\begin{proof}
By structural induction on $c$. First, we note that the construction of shadow
execution does not apply to sampling instructions.  Hence, $c$ and $\second{c}$ fall into the
deterministic portion of programs. Therefore, we write
$m_2=\evalexpr{c}{m_1}$ when $m_2$ is the unique memory such that
$\evalexpr{c}{m_1}(m_2) = 1$, and likewise for $c'$. Then this lemma
can be stated as 
\[m_2'=\evalexpr{\second{c}}{m_1'} \implies \second{\Gamma}m_2' = \evalexpr{c}{\second{\Gamma}m_1'}.\]  
\begin{itemize}[leftmargin=5mm]
    \item $c$ is $\skipcmd$: we have $m_1'=m_2'$ in this case. Hence,
$\second{\Gamma}m_2'=\second{\Gamma}m_1'=\evalexpr{\skipcmd}{\second{\Gamma}m_1'}$.

    \item $c$ is $(c_1;c_2)$: let $\shadowexec{c_1,\Gamma}=\second{c}_1$ and
$\shadowexec{c_2,\Gamma}=\second{c}_2$ and
$m'=\evalexpr{\second{c}_1}{m_1'}$. By
induction hypothesis, we have $\second{\Gamma} m' =
\evalexpr{c_1}{\second{\Gamma}m_1'}$. 
By language semantics,
$m_2'=\evalexpr{\second{c}_2}{m'}$. By induction hypothesis, with
$\second{\Gamma}$ as both the initial and final environments, we have
$\second{\Gamma} m_2'=\evalexpr{c_2}{\second{\Gamma}m'}=\evalexpr{c_1;c_2}{\second{\Gamma}m_1'}$

    \item $c$ is $x:=e$: we have $\second{c}=(\second{\distance{x}}:=
\shadowexec{e,\Gamma}-x)$ in this case. Moreover,
$m_2'=\subst{m_1'}{\second{\distance{x}}}{\evalexpr{\shadowexec{e,\Gamma})-x}{m_1'}}$.  
By Lemma \ref{lem:shadowexpr}, we have $\evalexpr{\shadowexec{e,\Gamma}}{m_1'}
= \evalexpr{e}{\second{\Gamma}m_1'}$.

For variable $x$, we know that $\second{\Gamma(x)}=*$ by
assumption. So
$\second{\Gamma} m_2'(x)=m_2'(x)+m_2'(\second{\distance{x}})=m_1'(x)+\evalexpr{\shadowexec{e,\Gamma_2}-x}{m_1'} =
\evalexpr{e}{\second{\Gamma} m_1'}
=\evalexpr{x:=e}{\second{\Gamma} m_1'}(x)$. For a variable $y$
other than $x$, the result is true since both $c$ and $\second{c}$ do not
modify $y$ and $\second{y}$, and $y$'s distances does not change due to the
well-formedness check in typing rule~\ruleref{T-Asgn}.

    \item $c$ is $\ifcmd{e}{c_1}{c_2}$: let $\shadowexec{c_1,\Gamma}=c_1'$
and $\shadowexec{c_2,\Gamma}=c_2'$. In this case, $\second{c} =
\ifcmd{\shadowexec{e,\Gamma}}{c_1'}{c_2'}$.  By Lemma \ref{lem:shadowexpr}, we
have $\evalexpr{\shadowexec{e,\Gamma}}{m_1} =
\evalexpr{e}{\second{\Gamma}m_1}$. Hence,
both $c$ and $\second{c}$ will take the same branch under
$\second{\Gamma}m_1'$ and $m_1'$ respectively.  The desired results follow from
induction hypothesis on $c_1$ or $c_2$.  

    \item $c$ is $\whilecmd{e}{c}$: again, since
$\evalexpr{\shadowexec{e,\Gamma}}{m_1} =
\evalexpr{e}{\second{\Gamma}m_1}$. The desired result
follows from induction on the number of loop iterations.
\end{itemize}
\end{proof}

Next, we show that when $\pc=\top$ (i.e., when the shadow execution might
diverge), the transformed code does not modify shadow variables and their
distances.

\begin{lemma}[High PC]
\label{lem:highpc}
$\forall c, c', \Gamma_1, \Gamma_2.~\top \proves
\cmdfull{\Gamma_1}{\Gamma_2}{c\transform c'}$, then we have
\begin{enumerate}[leftmargin=5mm]
\item $(\forall x\in \asgnvars(c).~\second{\Gamma_2}(x)=*) \land (\second\Gamma_1(x)=*\Rightarrow \second{\Gamma_2}(x)=*)$
\item $\forall m_1', m_2'.~\evalexpr{c'}{m_1'}(m_2') \not= 0 \implies \second\Gamma_1 m_1'=\second\Gamma_2 m_2'$
\end{enumerate}
\end{lemma}

\begin{proof}
By structural induction on $c$.

\begin{itemize}[leftmargin=5mm]
    \item $c$ is $\skipcmd$: trivial.

    \item $c$ is $x:=e$: When $\pc=\top$, $\second{\Gamma}_2(x)=*$ by the
typing rule. $(\second\Gamma_1(x)=*\Rightarrow \second{\Gamma_2}(x)=*)$ since
$\second\Gamma_2(y)=\second\Gamma_1(y)$ for $y\neq x$.

In this case, $c'$ is defined as $(\second{\distance{x}} := x + \second{\nexpr}
-e;x:=e)$ where $\second{\Gamma_1} \proves x: \second{\nexpr}$. By the semantics, we have
$m_2'=\subst{\subst{m_1'}{\second{\distance{x}}}{\evalexpr{x +
\second{\nexpr} -e}{m_1'}}}{x}{\evalexpr{e}{m_1'}}$. Hence, $\Gamma_2
m_2'(x)=m_2'(x)+\evalexpr{x + \second{\nexpr} -e}{m_1'} =
\evalexpr{e}{m_1'}+\evalexpr{x+\second\nexpr}{m_1'}-\evalexpr{e}{m_1'}=\Gamma_1
m_1(x)$.

For any $y\not=x$, both its value and its shadow distance do not change (due to
the wellformedness check in rule~\ruleref{T-Asgn}). Hence, $\second\Gamma_1
m_1'(y)=\second\Gamma_2 m_2'(y)$.

    \item $c$ is $c_1;c_2$: by induction hypothesis.

    \item $c$ is $\ifcmd{e}{c_1}{c_2}$: when $\pc=\top$, we have
$c'=\ifcmd{e}{(\cod{assert}~(\first{e});c_1';c_1'')}{(\cod{assert}~(\neg
\first{e});c_2'; c_2''})$. By the induction hypothesis, we know that $c_1'$ and
$c_2'$ does not modify any shadow variable and their ending typing
environments, say $\Gamma_1',\Gamma_2'$, satisfy condition 1.  Hence, the ending
environment $\Gamma_1'\join \Gamma_2'$ satisfies condition 1 too.

To show (2), we assume $c_1$ is executed without losing generality. Let $m_3'$
be the memory state between $c_1'$ and $c_1''$. By induction hypothesis,
$\second\Gamma_1 m_1'=\second{(\Gamma_1')} m_3'$. 
By the definition of $\Gamma_i, \Gamma_1\join \Gamma_2, \top \Rrightarrow
c_i''$, $c_1''$ and $c_2''$ only modifies $\first{\distance x}, x\in \Vars$.
Hence, $\second\Gamma_2 m_2'=\second{(\Gamma_1')} m_3'=\second\Gamma_1 m_1'$.

    \item $c$ is $\whilecmd{e}{c'}$: by the definition of $\join$ and induction
hypothesis, condition 1 holds.

Moreover, by the definition of $\Gamma, \Gamma\join \Gamma_f, \top \Rrightarrow
c_s$, $c_s$ only modifies $\first{\distance x}, x\in \Vars$. Hence,
$\second\Gamma_2 m_2'=\second\Gamma_1 m_1'$ by induction on the number of
iterations.

    \item $c$ is sampling instruction: does not apply since $\pc=\bot$ in the
typing rule.

\end{itemize}
\end{proof}

\begin{lemma}[]
\label{lem:nonsampling}
$\forall c, c', \second{c}, \Gamma_1, \Gamma_2.~(\forall x\in
\asgnvars(c).~\second{\Gamma_2}(x)=*) \land \shadowexec{c,\Gamma_2}=\second{c}
\land c'$ is deterministic, and 
$(\forall m_1',m_2'.~\evalexpr{c'}{m_1'}(m_2')\not=0 \implies 
\second\Gamma m_1' =\second\Gamma m_2')$, then we have
\[
\forall m_1',m_2'.~\evalexpr{c';\second{c}}{m_1'}(m_2') = \evalexpr{c}{\second{\Gamma_1}m_1'}(\second{\Gamma_2}m_2').
\]
\end{lemma}

\begin{proof}
$c$ and $c'$ are deterministic. Hence, this lemma can be stated as 
\[\forall m_1',m_2'.~m_2'=\evalexpr{c';\second{c}}{m_1'} \implies
\second{\Gamma_2}m_2' = \evalexpr{c}{\second{\Gamma_1}m_1'}\]

Let $m_3'= \evalexpr{c'}{m_1'}$ and $m_2'= \evalexpr{\second{c}}{m_3'}$. By
assumption, we have $\second\Gamma_2 m_3'=\second\Gamma_1 m_1'$. Moreover, by
Lemma~\ref{lem:shadowexec}, we have $\second\Gamma_2
m_2'=\evalexpr{c}{\second\Gamma_2 m_3'}$. Hence, $\second\Gamma_2
m_2'=\evalexpr{c}{\second\Gamma_2 m_3'}=\evalexpr{c}{\second\Gamma_1 m_1'}$.

\end{proof}

\subsection{Soundness}
Finally, we prove the Pointwise Soundness Lemma.

\paragraph{Proof of Lemma~\ref{lem:pointwise}}
\begin{proof}
By structural induction on $c$. Note that the desired inequalities are trivially true if $\evalexpr{c'}{m_1'}(m_2') = 0$. Hence, in the proof we assume that $\evalexpr{c'}{m_1'}(m_2')> 0$. 

\begin{itemize}[leftmargin=5mm]
\item %
Case $(\skipcmd)$: trivial.

\item %
Case $(x:=e)$: 
by \ruleref{T-Asgn}, we have $\max(\restrict{c''}{m_1'}{m_2'})=0$. This command is deterministic in the sense that when
$\evalexpr{c'}{m_1'}(m_2') \neq 0$, we have
$\evalexpr{c'}{m_1'}(m_2') = 1$ and $m_2' =
\subst{m_1'}{x}{\evalexpr{e}{m_1'}}$.  To prove \eqref{inv:shadow} and
\eqref{inv:aligned}, it suffices to show that $\second{\Gamma}_2 m_2'
=\subst{\second{\Gamma}_1 m_1'}{x}{\evalexpr{e}{\second{\Gamma_1} m_1'}}$ and $\first{\Gamma}_2 m_2'
=\subst{\first{\Gamma}_1 m_1'}{x}{\evalexpr{e}{\first{\Gamma}_1 m_1'}}$.  We prove the latter
one as the other (given $\pc=\bot$) can be shown by a similar argument.

First, we show 
$\first{\Gamma}_2 m_2'(x) =\subst{\first{\Gamma}_1 m_1'}{x}{\evalexpr{e}{\first{\Gamma}_1 m_1'}}(x)$.
Let $\Gamma_1 \proves e: \pair{\first{\nexpr}}{\second{\nexpr}}$. By the typing rule, we have $\first{\Gamma_2}(x) = \first{\nexpr}$, and 
\begin{align*}
\first{\Gamma}_2 m_2' (x) &= m_2'(x) + \evalexpr{\first{\nexpr}}{m_2'} \\
&=\evalexpr{e}{m_1'} + \evalexpr{\first{\nexpr}}{m_2'} \\
&=\evalexpr{e}{m_1'} + \evalexpr{\first{\nexpr}}{m_1'} \\
&=\evalexpr{e}{\first{\Gamma}_1 m_1'}\\
&= \subst{\first{\Gamma}_1 m_1'}{x}{\evalexpr{e}{\first{\Gamma}_1 m_1'}} (x).
\end{align*}
The third equality is due to well-formedness, and the forth equality is due to
Lemma~\ref{lem:aexpr}.

Second, we show that $\first{\Gamma}_2 m_2' (y) =\subst{\first{\Gamma}_1
m_1'}{x}{\evalexpr{e}{\first{\Gamma}_1 m_1'}} (y)$ for $y\neq x$. First, by
Rule~\ruleref{T-Asgn}, $\first\Gamma_2(y)=\first\Gamma_1(y)$, $m_1'(y)=m_2'(y)$
and $m_1'(\first{\distance{y}})=m_2'(\first{\distance{y}})$. If
$\first\Gamma_2(y)=\first\Gamma_1(y)=*$, $\first{\Gamma}_2 m_2'
(y)=m_2'(y)+m_2'(\first{\distance y})=m_1'(y)+m_1'(\first{\distance y})=\first{\Gamma}_1 m_1'(y)$. If
$\first\Gamma_2(y)=\first\Gamma_1(y)=\nexpr$, $\first{\Gamma}_2 m_2'
(y)=m_2'(y)+\evalexpr{\nexpr}{m_2'}=m_1'(y)+\evalexpr{\nexpr}{m_1'}=\first{\Gamma}_1
m_1'(y)$, where $\evalexpr{\nexpr}{m_2'}=\evalexpr{\nexpr}{m_1'}$ due to
well-formedness.

\item %
Case $(c_1;c_2)$: For any $m_2'$ such that $\evalexpr{c_1';c_2'}{m_1'}(m_2')\neq 0$,
there exists some $m'$ such that 
\[\evalexpr{c_1'}{m_1'}(m')\not=0 \land \evalexpr{c_2'}{m'}(m_2')\not=0.\] 
Let $\pc \proves \Gamma_1 \{c_1 \transform c_1' \} \Gamma$, $\pc \proves \Gamma
\{c_2 \transform c_2'\} \Gamma_2$ For \eqref{inv:shadow}, we have
\begin{align*}
\evalexpr{c_1';c_2'}{m_1'}(m_2') &= \sum_{m'} \evalexpr{c_1'}{m_1'}(m')\cdot \evalexpr{c_2'}{m'}(m_2') \\
&\leq \sum_{m'} \evalexpr{c_1}{\second{\Gamma}_1 m_1'}(\second{\Gamma} m')\cdot \evalexpr{c_2}{\second{\Gamma} m'}(\second{\Gamma}_2 m_2)\\
&\leq \sum_{m'} \evalexpr{c_1}{\second{\Gamma}_1 m_1'}(m')\cdot \evalexpr{c_2}{m'}(\second{\Gamma}_2 m_2')\\
&= \evalexpr{c_1;c_2}{\second{\Gamma}_1 m_1'}(\second{\Gamma}_2 m_2').
\end{align*}
Here the second line is by induction hypothesis. The change of variable in the
third line is due to Lemma~\ref{lem:injective}.

For \eqref{inv:aligned} and \eqref{inv:aligned2}, let $\priv_1=\max
(\restrict{c_1''}{m_1'}{m'})$, $\priv_2=\max (\restrict{c_2''}{m'}{m_2'})$ and
$\priv = \max (\restrict{c_1'';c_2''}{m_1'}{m_2'})$. Note that $\priv_1+\priv_2
\leq \priv$ due to the fact that $\extmemfull{m_2'}{\priv_1+\priv_2}{0}\in
\evalexpr{c_1'';c_2''}{\extmemfull{m_1'}{0}{0}}$. By induction hypothesis we
have one of the following two cases. 
\begin{enumerate}
    \item $\evalexpr{c_2'}{m'}(m_2')\leq\exp(\priv_2)\evalexpr{c_2}{\second{\Gamma} m'}(\first{\Gamma}_2 m_2')$. We have
\begin{align*}
\evalexpr{c_1';c_2'}{m_1'}(m_2') &= \sum_{m'} \evalexpr{c_1'}{m_1'}(m')\cdot \evalexpr{c_2'}{m'}(m_2') \\
&\leq \exp(\priv_2) \sum_{m'} \evalexpr{c_1'}{\second{\Gamma}_1 m_1'}(\second{\Gamma} m')\cdot \evalexpr{c_2'}{\second{\Gamma} m'}(\first{\Gamma}_2 m_2')\\
&\leq \exp(\priv_2) \sum_{m'} \evalexpr{c_1'}{\second{\Gamma}_1 m_1'}(m')\cdot \evalexpr{c_2'}{m'}(\first{\Gamma}_2 m_2')\\
&\leq \exp(\priv_2) \evalexpr{c_1;c_2}{\second{\Gamma}_1 m_1'}(\first{\Gamma}_2 m_2')\\
&\leq \exp(\priv) \evalexpr{c_1;c_2}{\second{\Gamma}_1 m_1'}(\first{\Gamma}_2 m_2').
\end{align*}
The first inequality is by induction hypothesis. The change
of variable in the third line is again due to Lemma~\ref{lem:injective}. The last line is because $\priv_1+\priv_2\leq \priv$.

\item $\evalexpr{c_2'}{m'}(m_2')\leq\exp(\priv_2)\evalexpr{c_2}{\first{\Gamma} m'}(\first{\Gamma}_2 m_2')$. By induction hypothesis on $c_1$, we have two more cases:
\begin{enumerate}
    \item $\evalexpr{c_1'}{m_1'}(m')\leq\exp(\priv_1)\evalexpr{c_1}{\first{m_1}}(\first{m})$. In this case,
    \begin{align*}
\evalexpr{c_1';c_2'}{m_1'}(m_2') &= \sum_{m'} \evalexpr{c_1'}{m_1'}(m')\cdot \evalexpr{c_2'}{m'}(m_2') \\
&\leq \exp(\priv) \sum_{m'} \evalexpr{c_1}{\first{\Gamma}_1 m_1'}(\first{\Gamma} m')\cdot \evalexpr{c_2}{\first{\Gamma} m'}(\first{\Gamma}_2 m_2')\\
&\leq \exp(\priv) \sum_{m'} \evalexpr{c_1}{\first{\Gamma}_1 m_1'}(m')\cdot \evalexpr{c_2}{m'}(\first{\Gamma}_2 m_2')\\
&\leq \exp(\priv) \evalexpr{c_1;c_2}{\first{\Gamma}_1 m_1'}(\first{\Gamma}_2 m_2').
\end{align*}
The second line is by induction hypothesis and the fact that $\priv_1+\priv_2\leq \priv$. The change
of variable in the third line is due to Lemma~\ref{lem:injective}.

\item $\evalexpr{c_1'}{m_1'}(m')\leq\exp(\priv_1)\evalexpr{c_1}{\second{\Gamma}_1 m_1'}(\first{\Gamma} m')$. We have 
 \[\evalexpr{c_1';c_2'}{m_1'}(m_2') \leq \exp(\priv) \evalexpr{c_1;c_2}{\second{\Gamma}_1 m_1'}(\first{\Gamma}_2 m_2')\]
 by a similar argument as above.
\end{enumerate}
\end{enumerate}

\item %
Case $(\ifcmd{e}{c_1}{c_2})$: 
If $\Gamma_1 \proves e:\bool$, then by
Lemma~\ref{lem:bexpr} we have $\evalexpr{e}{m_1'}=\evalexpr{e}{\first{\Gamma}_1 m_1'} =
\evalexpr{e}{\second{\Gamma}_1 m_1'}$.  Hence, the same branch is taken in all related
executions. By rule~\ruleref{T-If}, the transformed program is
\[    \cod{if}~{e}~\cod{then}~{(\cod{assert}~(\alignexec{e});c_1';c_1'')}%
    ~\cod{else}~{(\cod{assert}~(\neg \alignexec{e});c_2'; c_2''})\]
and 
$\Gamma_1, \Gamma_1\join\Gamma_2,\bot \Rrightarrow c_1'',\quad \Gamma_2,
\Gamma_1\join\Gamma_2,\bot \Rrightarrow c_2''$.  Without loss of generality,
suppose that $c_1$ is executed in all related executions. By
Lemma~\ref{lem:triplearrow}, there is a unique $m'$ such that
$\evalexpr{c_1'}{m_1'}(m')\not=0$. By induction hypothesis, we have
\[
\evalexpr{c_1'}{\store_1'}(\store') \leq
\evalexpr{c_1}{\second{\Gamma}\store_1'}(\second{\Gamma_1}\store') \]

Moreover, by Lemma~\ref{lem:triplearrow}, $\second{\Gamma_1}\store'=\second{\Gamma_1}\store_2'$ and
\[ \evalexpr{c_1';c_1''}{m_1'}(m_2')=
\evalexpr{c_1'}{m_1'}(m').\]

Hence, we have \[\evalexpr{c_1';c_1''}{m_1'} m_2' \leq
\evalexpr{c_1}{\second{\Gamma}\store_1'}(\second{\Gamma_1}\store')  =
\evalexpr{c_1}{\second{\Gamma}\store_1'}(\second{\Gamma_1}\store_2') \]

\eqref{inv:aligned} or \eqref{inv:aligned2} can be proved in a
similar way.

When $\Gamma_1\not\proves e:\bool$, the aligned execution will still take the
same branch due to the inserted assertions. Hence, the argument above still
holds for \eqref{inv:aligned} or \eqref{inv:aligned2}.

For the shadow execution, we need to prove \eqref{inv:shadow}. By
rule~\ruleref{T-If}, $\pc' = \top$, and the transformed program is
\begin{align*}
    &\ifcmd{e}{(\cod{assert}~(\alignexec{e});c_1';c_1'')}{(\cod{assert}~(\neg \alignexec{e});c_2'; c_2''});\\
    &\shadowexec{\ifcmd{e}{c_1}{c_2},\Gamma_1\join\Gamma_2}
\end{align*}
By Lemma~\ref{lem:highpc} and the definition of $\Rrightarrow$ under high pc, we
have that $\forall m_1',m_2'.~\evalexpr{c'}{m_1'}(m_2')\neq 0 \Rightarrow
\second \Gamma_2 m_1'=\second \Gamma_1 m_1'$, and
$\Gamma_1\join\Gamma_2(x)=*~\forall x\in \asgnvars(c_1;c_2)$. Furthermore, the
program is deterministic since it type-checks under $\top$. Therefore,
$\max(\restrict{W''}{m_1}{m_2})=0$ and \eqref{inv:shadow} holds by
Lemma~\ref{lem:nonsampling}.

\item %
Case $(\whilecmd{e}{c})$: 
Let $W=\whilecmd{e}{c}$. If $\Gamma_1\proves
e:\bool$, then $\evalexpr{e}{m_1'}=\evalexpr{e}{\first\Gamma_2 m_1'}=\evalexpr{e}{\second{\Gamma_2} m_1'}$. By rule \ruleref{T-While} we have $\pc'=\bot$, $\Gamma_2=\Gamma\join \Gamma_f$ and the transformed program is
\[ W'=c_s; ~\whilecmd{e}{(\cod{assert}~(\alignexec{e});c';c'')}\] where $\bot\proves
\Gamma_1\join\Gamma_f~\{c \transform c'\}~\Gamma_f$, $(\Gamma_1,
\Gamma_1\join\Gamma_f, \bot)\Rrightarrow c_s$ and $(\Gamma_f,
\Gamma_1\join\Gamma_f, \bot)\Rrightarrow c''$.  We proceed by natural induction
on the number of loop iterations (denoted by $i$).

When $i=0$, we have $\evalexpr{e}{m_1'}=\false$. By the semantics, we have $\evalexpr{W'}{m_1'}=\unitop~(\evalexpr{c_s}{m_1'})$,
$\evalexpr{W}{\first\Gamma_1 m_1'}=\unitop~(\first\Gamma_1 m_1')$, and
$\evalexpr{W}{\second\Gamma_1 m_1'}=\unitop~(\second\Gamma_1 m_1')$.
Furthermore, $\max (\restrict{W''}{m_1}{m_2})=0$. By
Lemma~\ref{lem:triplearrow}, $(\first\Gamma_1 m_1')=(\first\Gamma_2
(\evalexpr{c_s}{m_1'}))$ and $(\second\Gamma_1 m_1')=(\second\Gamma_2
(\evalexpr{c_s}{m_1'}))$. So desired result holds.

When $i=j+1>0$, we have $\evalexpr{e}{m_1'}=\true$. By the semantics, we have
$\evalexpr{W'}{m_1'}=\evalexpr{c_s;\underbrace{c';c'';}_{j}c';c''}{m_1'}$ and
$\evalexpr{W}{\Gamma_1^\star
m_1'}=\evalexpr{c_s;\underbrace{c';c'';}_{j}c';c''}{\Gamma_1^\star
m_1'},\star\in \{\alignd,\shadow\}$. For any $m'$ such that
$\evalexpr{c_s;\underbrace{c';c'';}_{j}}{m_1'}(m')=k_1\not=0$ and
$\evalexpr{c';c''}{m'}(m_2')=k_2\not=0$, we know that
$\evalexpr{\underbrace{c}_{j}}{\Gamma_1^\star m_1'}(\Gamma_2^\star m')=k_1$.
Moreover, there is a unique $m_3'$ such that $\evalexpr{c'}{m'}(m_3')=k_2\land
\evalexpr{c''}{m_3'}(m_2')=1$ by Lemma~\ref{lem:triplearrow}. By structural
induction on $c$, we have $k_1\leq \evalexpr{c}{\second\Gamma_2 m'}(\second\Gamma_f
m_3')$ and either $k_1 \leq \exp(\max(\restrict{c'}{m_1'}{m_2'}))
\evalexpr{c}{\first\Gamma_2 m'}(\first\Gamma_f m_3')$ or $k_1 \leq
\exp(\max(\restrict{c''}{m_1'}{m_2'})) \evalexpr{c}{\second\Gamma_2
m'}(\first\Gamma_f m_3')$. By Lemma~\ref{lem:triplearrow}, $\Gamma_f^\star
m_3'=\Gamma_2^\star m_2'$. Hence, we can replace $\Gamma_f$ with $\Gamma_2$ in
the inequalities. Finally, the desired result holds following the same argument
as in the composition case.

Otherwise, by rule~\ruleref{T-While} we have $\pc' = \top$, $\Gamma_2=\Gamma
\join \Gamma_f$ and the transformed program is $W'=c_s;W''$ where \[ W''=
\whilecmd{e}{(\cod{assert}~(\alignexec{e});c';c'')};~\shadowexec{\whilecmd{e}{c},\Gamma_2}\]
and $\top\proves \Gamma_1\join\Gamma_f~\{c \transform c'\}~\Gamma_f$,
$(\Gamma_1, \Gamma_1\join\Gamma_f, \top)\Rrightarrow c_s$ and
$(\Gamma_f,\Gamma\join\Gamma_f,\top)\Rrightarrow c''$.  Since we do not have
sampling instructions in this case, the program is deterministic, thus
$\max(\restrict{W''}{m_1}{m_2})=0$. Note that by rule~\ruleref{T-Asgn},
$\forall x\in \asgnvars(c)$, $\second\Gamma_f(x)=*$. Hence,
$\second\Gamma_2(x)=*$ for $x\in \asgnvars(c)$. Morevoer, let $m'$ be the
unique memory such that $\evalexpr{c_s}{m_1'}(m')\not=0$, we have $\second\Gamma_1
m_1'=\second\Gamma_2 m'$ by Lemma~\ref{lem:triplearrow}.  Therefore, by
Lemmas~\ref{lem:highpc} and~\ref{lem:nonsampling}, we have \[\evalexpr{W'}{\store_1'}(\store_2') =
\evalexpr{W''}{\store'}(\store_2')=
\evalexpr{W}{\second{\Gamma_2}\store'}(\second{\Gamma_2}\store_2')=
\evalexpr{W}{\second{\Gamma_1}\store_1'}(\second{\Gamma_2}\store_2')\]

\item %
Case $(\eta := \lapm~\real;\select,\nexpr_\eta )$: let $\mu_r$ be the probability  density (resp. mass) function of the  Laplace (resp. discrete Laplace) distribution with zero mean and scale $r$. Since $\mu_r(v)\propto \exp(-|v|/r)$ we have
\begin{equation}\label{eq:lapr}
\forall v,d\in \mathbb{R},~\mu_r(v) \leq \exp(|d|/ r) \mu_r (v+d).
\end{equation}
When $\evalexpr{c'}{m_1'}(m_2') \neq 0$, we have $m_2'=\subst{m_1'}{\eta
}{\cod{v}}$ for some constant $\cod{v}$, and $\evalexpr{c'}{m_1'}(m_2') =
\mu_r(\cod{v}).$

We first consider the case $\select = \circ$. By \ruleref{T-Laplace}, we have ${\Gamma_2}(x) = {\Gamma_1}(x)$ for $x\neq \eta$, and ${\Gamma_2}(\eta) = \pair{\nexpr_\eta}{0}$. Therefore, for $x\neq \eta$, we have
\begin{align*}
\second\Gamma_2 {m_2'} (x) &= m_2'(x) + \evalexpr{\second\Gamma_2 (x)}{m_2'}\\
&=m_1'(x) + \evalexpr{\second{\Gamma_1}(x)}{m_2'} \\
&=m_1'(x) + \evalexpr{\second{\Gamma_1}(x)}{m_1'} \\
&= \second\Gamma_1 {m_1'} (x).
\end{align*}

In the second equation, we have $m_2'(x) = m_1'(x)$ because $m_2=\subst{m_1}{\eta
}{\cod{v}}$ and $x\neq \eta$. The third equation is due to the well-formedness
assumption.  Also, $\second\Gamma_2 {m_2'} (\eta)= m_2' (\eta) = \cod{v}$ since
$\second{\Gamma_2}(\eta) = 0$. Therefore, we have $\second\Gamma_2 {m_2'} =
\subst{\second\Gamma_1 {m_1'}}{\eta}{\cod{v}}$ and thus
$\evalexpr{c}{\second\Gamma_1 {m_1'}}(\second\Gamma_2 {m_2'}) =
\mu_r(\cod{v})=\evalexpr{c'}{m_1'}(m_2')$.

Similarly, we can show that $\first\Gamma_2 {m_2'} = \subst{\first\Gamma_1
{m_1'}}{\eta}{\cod{v}+\cod{d}}$ and therefore $\evalexpr{c}{\first\Gamma_1
{m_1'}}(\first\Gamma_2
{m_2'})=\mu_r(\cod{v}+\cod{d})$. Furthermore, since
$\vpriv{\priv} :=\vpriv{\priv}+|\nexpr_\eta|/r$, the set
$(\restrict{c''}{m_1}{m_2})$ contains a single element
$\evalexpr{|\nexpr_\eta|/ r}{m_2} =|\cod{d}|/ r$, and thus
$\max(\restrict{c''}{m_1}{m_2}) =|\cod{d}|/ r$. Therefore, we have
\begin{align*}
\evalexpr{c'}{m_1'}(m_2') &= \mu_r (\cod{v})\\
&\leq\exp(|\cod{d}|/ r) \mu_r(\cod{v}+\cod{d})\\
&=\exp(\max(\restrict{c''}{m_1}{m_2}))\evalexpr{c}{\first\Gamma_1 {m_1}}(\first\Gamma_2 {m_2}).
\end{align*}
The second inequality is due to \eqref{eq:lapr}. Thus \eqref{inv:shadow} and \eqref{inv:aligned} hold.

Next consider the case $\select = \dagger$. By the typing rule, we have
$\first{\Gamma_2}(x) = \second{\Gamma_2}(x) = \second{\Gamma_1}(x)$ for $x\neq
\eta$, and ${\Gamma_2}(\eta) = \pair{\nexpr_\eta}{0}$. By a similar argument we
have $\second\Gamma_2 {m_2'} = \subst{\second\Gamma_1 {m_1'}}{\eta}{\cod{v}}$
and $\first\Gamma_2 {m_2'} = \subst{\second\Gamma_1
{m_1'}}{\eta}{\cod{v}+\cod{d}}$. Furthermore, since $\vpriv{\priv}
:=0+|\nexpr_\eta|/r$, we have $\max(\restrict{c''}{m_1'}{m_2'}) =|\cod{d}|/ r$.
Therefore, we have 
\begin{align*}
\evalexpr{c'}{m_1'}(m_2') &= \mu_r (\cod{v})
=\evalexpr{c}{\second\Gamma_1 {m_1'}}(\second\Gamma_2 {m_2'}) \\
&\leq\exp(|\cod{d}|/ r) \mu_r(\cod{v}+\cod{d})\\
&=\exp(\max(\restrict{c''}{m_1'}{m_2'}))\evalexpr{c}{\second\Gamma_1 {m_1'}}(\first\Gamma_2 {m_2'}).
\end{align*}
Thus \eqref{inv:shadow} and \eqref{inv:aligned2} hold.

Finally, consider $\select = e\ ?\ \select_1:\select_2$. For any $m_1'$, we have $\select = \select_1$ if $\evalexpr{e}{m_1'} = \true$ and $\select = \select_2$ if $\evalexpr{e}{m_1'} = \false$. By induction, $\select$ evaluates to either $\circ$ or $\dagger$ under $m_1'$. Thus the proof follows from the the two base cases above.
\end{itemize}
\end{proof}

\section{Formal Semantics for the Target Language}

The denotational semantics interprets a command $c$ in the target language
as a function $\trans{c}:
\Store\rightarrow \mathcal{P}(\Store)$.  The semantics of commands are formalized as
follows. 
\begin{align*}
\trans{\skipcmd}_\store &= \{m\} \\
\trans{x:=e}_\store &= \{\subst{m}{x}{\evalexpr{e}{m}}\} \\
\trans{\havoc{x}}_\store &= \cup_{r\in \mathbb{R}}~\{\subst{m}{x}{r}\} \\
\trans{c_1;c_2}_\store &= \cup_{m'\in \evalexpr{c_1}{\store}} \evalexpr{c_2}{m'}  \\
\trans{\ifcmd{e}{c_1}{c_2}}_\store &= 
      \begin{cases} 
      \trans{c_1}_\store &\mbox{if } \evalexpr{e}{\store} = \true \\
      \trans{c_2}_\store &\mbox{if } \evalexpr{e}{\store} = \false \\
      \end{cases} \\
\trans{\whilecmd{e}{c}}_\store &= w^*~m\\
\text{where } w^*              &= fix (\lambda f.~\lambda m. \cod{if}~\evalexpr{e}{m}=\true\\
                               &\qquad \cod{then}~(\cup_{m'\in \evalexpr{c}{m}}~f~m')~\cod{else}~{\{m\}}) \\
\trans{c;\outcmd{e}}_\store &= \cup_{m'\in \evalexpr{c}{\store}}~\{\evalexpr{e}{m'}\}
\end{align*}

%% file: semantics.tex
\begin{figure}[ht]
\setstretch{0.9}
\begin{align*}
\trans{\skipcmd}_\store &= \unitop~m \\
\trans{x:=e}_\store &= \unitop~(\subst{m}{x}{\evalexpr{e}{m}}) \\
\trans{\eta:=g,\select,\nexpr_{\eta}}_\store &= \bindop~\evalexpr{g}{}~(\lambda v.~\unitop~\subst{m}{\eta }{v}) \\
\trans{c_1;c_2}_\store &= \bindop~(\evalexpr{c_1}{\store})~\trans{c_2}  \\
\trans{\ifcmd{e}{c_1}{c_2}}_\store &= 
      \begin{cases} 
      \trans{c_1}_\store &\mbox{if } \evalexpr{e}{\store} = \true \\
      \trans{c_2}_\store &\mbox{if } \evalexpr{e}{\store} = \false \\
      \end{cases} \\
\trans{\whilecmd{e}{c}}_\store &= w^*~m\\
\text{where } w^*              &= fix (\lambda f.~\lambda m. \cod{if}~\evalexpr{e}{m}=\true\\
                               \cod{then}~&(\bindop~\evalexpr{c}{m}~f)~\cod{else}~{(\unitop~m)}) \\
\trans{c;\outcmd{e}}_\store &= \bindop~(\evalexpr{c}{\store})~(\lambda m'.~\unitop~\evalexpr{e}{m'})
\end{align*}
\caption{\lang: language semantics.}
\label{fig:semantics}
\end{figure}

%% file: shadowrules.tex
\begin{figure}[ht]
\setstretch{0.85}
\raggedright
\begin{mathpar}
\starexec{\real,\Gamma} = \real  
\quad
\starexec{\true,\Gamma} = \true
\quad
\starexec{\false,\Gamma} = \false
\and 
\starexec{x,\Gamma} = 
\begin{cases}
x+\third\nexpr &\text{, else if } \Gamma \proves x : \tyreal_{\pair {\first\nexpr} {\second\nexpr}} \cr
x &\text{, else } 
\end{cases}
\and
\starexec{e_1~\op~e_2,\Gamma} = \starexec{e_1,\Gamma}~\op~\starexec{e_2,\Gamma}
\text{ where } \op = \oplus\cup \otimes \cup \odot
\\
\starexec{e_1[e_2],\Gamma} =~
\begin{cases}
e_1[e_2]+\third{\distance{e_1}}[e_2] &\text{, if } \third{\Gamma} \proves e_1 : \tylist~\tyreal_* \cr
e_1[e_2] + \third{\nexpr} &\text{, else if } \third{\Gamma} \proves e_1 : \tylist~\tyreal_{\third{\nexpr}} \cr
e_1[e_2] &\text{, else} 
\end{cases}
\and 
\starexec{e_1::e_2,\Gamma} = 
\starexec{e_1,\Gamma}::\starexec{e_2,\Gamma} 
\and
\starexec{\neg e,\Gamma} = \neg \starexec{e, \Gamma}
\and
\starexec{e_1\mathbin{?}e_2:e_3,\Gamma} = \starexec{e_1}\mathbin{?}\starexec{e_2,\Gamma}:\starexec{e_3,\Gamma}
\end{mathpar}
\caption{Transformation of numerical expressions for aligned and shadow execution, where $\star \in \{\alignd, \shadow\}$.}
\label{fig:shadowrules_expression}
\end{figure}

\begin{figure}[ht]
\setstretch{0.85}
\raggedright
\begin{mathpar}
\shadowexec{\skipcmd,\Gamma} = \skipcmd 
\and
\inferrule{
\shadowexec{c_1;\Gamma} = c_1'
\quad 
\shadowexec{c_2;\Gamma} = c_2'
}{
\shadowexec{c_1;c_2,\Gamma} = c_1';c_2'
}
\and
\shadowexec{x:=e,\Gamma} = ({\second{\distance{x}}} := \shadowexec{e,\Gamma}-x)
\and
\inferrule{
\shadowexec{c_i, \Gamma} = c_i' \quad i \in \{1,2\}
}
{
\shadowexec{\ifcmd{e}{c_1}{c_2},\Gamma} = \ifcmd{\shadowexec{e,\Gamma}}{c_1'}{c_2'}}
\and
\inferrule{
\shadowexec{c,\Gamma}= c'
}{\shadowexec{\whilecmd{e}{c},\Gamma} = \whilecmd{\shadowexec{e,\Gamma}}{c'}}
\end{mathpar}
\caption{Shadow execution for commands.}
\label{fig:shadowrules}
\end{figure}